\documentclass[11pt]{article}

\usepackage{amssymb,amsmath,amsfonts,amsthm}
\usepackage{mathtools}
\usepackage{dsfont}
\usepackage{mathpazo}
\usepackage{fullpage}

\usepackage{authblk}

\usepackage{titlesec}
\titleformat{\section}[hang]{\Large\bfseries\filright}{\thesection.}{.5em}{}
\titleformat{\subsection}[hang]{\large\bfseries\filright}{}{0em}{}
\titleformat{\subsubsection}[block]{\bfseries}{}{0em}{}

\usepackage[colorlinks=true,linkcolor=blue,
  citecolor=blue,urlcolor=blue]{hyperref}

\newcommand{\Tr}{\operatorname{Tr}}
\newcommand{\I}{\mathds{1}}
\renewcommand{\t}{{\scriptscriptstyle\mathsf{T}}}
\renewcommand\vec{\operatorname{vec}}
\newcommand{\complex}{\mathbb{C}}
\newcommand{\M}{\mathcal{M}}
\newcommand{\U}{\mathcal{U}}
\newcommand{\op}{\operatorname}

\newcommand{\ip}[2]{\langle #1 , #2\rangle}
\newcommand{\bigip}[2]{\bigl\langle #1, #2 \bigr\rangle}

\newcommand{\abs}[1]{\lvert #1 \rvert}
\newcommand{\bigabs}[1]{\bigl\lvert #1 \bigr\rvert}

\theoremstyle{plain}
\newtheorem{theorem}{Theorem}
\newtheorem{lemma}[theorem]{Lemma}
\newtheorem{corollary}[theorem]{Corollary}
\newtheorem{proposition}[theorem]{Proposition}

\theoremstyle{definition}

\newtheorem{example}[theorem]{Example}

\theoremstyle{remark}
\newtheorem*{remark}{Remark}

\begin{document}

\title{\textbf{On the mixed-unitary rank of quantum channels}}

\author[1,2]{Mark Girard}
\author[1,3,4]{Debbie Leung}
\author[1,5]{Jeremy Levick}
\author[1,6]{Chi-Kwong Li}
\author[1,7]{\\Vern Paulsen}
\author[8,9]{Yiu Tung Poon}
\author[1,2,10]{John Watrous\vspace{-2mm}}

\affil[1]{Institute for Quantum Computing, University of Waterloo\vspace{0.4mm}}

\affil[2]{School of Computer Science, University of Waterloo\vspace{0.4mm}}

\affil[3]{Department of Combinatorics \& Optimization, University of
  Waterloo\vspace{0.4mm}}

\affil[4]{Perimeter Institute for Theoretical Physics, Waterloo\vspace{0.4mm}}

\affil[5]{Department of Mathematics \& Statistics, University of
  Guelph\vspace{0.4mm}}

\affil[6]{Department of Mathematics, College of William and Mary\vspace{0.4mm}}

\affil[7]{Department of Pure Mathematics, University of Waterloo\vspace{0.4mm}}

\affil[8]{Department of Mathematics, Iowa State University\vspace{0.4mm}}

\affil[9]{Center for Quantum Computing, Peng Cheng Laboratory\vspace{0.4mm}}

\affil[10]{Canadian Institute for Advanced Research, Toronto\vspace{-1mm}}

\date{\today}

\maketitle

\begin{abstract}
  In the theory of quantum information, the
  \emph{mixed-unitary quantum channels}, for any positive integer dimension
  $n$, are those linear maps that can be expressed as a convex combination of
  conjugations by $n\times n$ complex unitary matrices.
  We consider the \emph{mixed-unitary rank} of any such channel, which is the
  minimum number of distinct unitary conjugations required for an expression of
  this form.
  We identify several new relationships between the mixed-unitary rank~$N$ and
  the Choi rank~$r$ of mixed-unitary channels, the Choi rank being equal to the
  minimum number of nonzero terms required for a Kraus representation of that
  channel.
  Most notably, we prove that the inequality $N\leq r^2-r+1$ is satisfied for
  every mixed-unitary channel (as is the equality $N=2$ when $r=2$), and
  we exhibit the first known examples of mixed-unitary channels for which
  $N>r$.
  Specifically, we prove that there exist mixed-unitary channels having Choi
  rank $d+1$ and mixed-unitary rank $2d$ for infinitely many positive integers
  $d$, including every prime power $d$.
  We also examine the mixed-unitary ranks of the mixed-unitary Werner--Holevo
  channels.
\end{abstract}

\section{Introduction}

The theory of quantum information posits that discrete-time changes in
quantum-mechanical systems that store quantum information are represented
by \emph{quantum channels} (or just \emph{channels} for short), which are
completely positive and trace-preserving linear maps from square matrices to
square matrices having complex number entries.
Hereafter we shall write $\M_{n,m}$ to denote the space of $n\times m$ complex
matrices, and we also write $\M_n = \M_{n,n}$.

One standard way of describing any given channel $\Phi:\M_n\rightarrow\M_m$
is to choose a positive integer $r$ along with matrices
$A_1,\ldots,A_r\in M_{m,n}$, and then take
\begin{equation}
  \label{eq:Kraus form}
  \Phi(X) = \sum_{k=1}^r A_k X A_k^{\ast}
\end{equation}
for all $X\in\M_n$.
Such a description is known as a \emph{Kraus representation} of $\Phi$, and the
existence of such a description is equivalent to $\Phi$ being completely
positive \cite{Choi1975}.
A map $\Phi$ described in this way preserves trace if and only if
\begin{equation}
  \sum_{k=1}^r A_k^{\ast} A_k = \I_n,
\end{equation}
where $\I_n\in\M_n$ is the identity matrix.
The minimum value of $r$ for which such a description exists is called the
\emph{Choi rank} of $\Phi$, this number being so-named because it is equal to
the rank of the \emph{Choi representation} (or \emph{Choi matrix}) $J(\Phi)$
associated with $\Phi$:
\begin{equation}
  J(\Phi) = \sum_{1\leq j,k\leq n} \Phi(E_{j,k}) \otimes E_{j,k},
\end{equation}
where $E_{j,k}\in\M_n$ denotes the matrix having a 1 in entry $(j,k)$
and 0 in all other entries.
As the Choi representation is an $nm\times nm$ matrix, it follows that the
Choi rank of $\Phi$ is always at most $nm$.

In this paper we consider a restricted class of channels called mixed-unitary
channels.
Writing $\U_n$ to denote the set of all $n\times n$ complex unitary matrices,
we say that a channel $\Phi:\M_n\rightarrow\M_n$ is a \emph{unitary channel} if
there exists a unitary matrix $U\in\U_n$ for which
\begin{equation}
  \Phi(X) = U X U^{\ast}
\end{equation}
for all $X\in\M_n$, and we say that a channel is a
\emph{mixed-unitary channel} if it can be expressed as a convex combination of
unitary channels.
That is, a channel $\Phi:\M_n\rightarrow\M_n$ is mixed unitary if and only if
there exists a positive integer $N$, a probability vector $(p_1,\ldots,p_N)$,
and unitary matrices $U_1,\ldots,U_N\in\U_n$ such that
\begin{equation}
  \label{eq:general mixed-unitary channel}
  \Phi(X) = \sum_{k = 1}^N p_k U_k X U_k^{\ast}
\end{equation}
for all $X\in\M_n$. We observe that the set of all mixed-unitary channels is compact, as it is the convex hull of a compact set (the set of unitary channels) in a finite-dimensional space. Thus, by Carath\'eodory's theorem, every element in the closed convex hull of the set of unitary channels can be represented as a (finite) convex combination of unitary channels.

Various properties of mixed-unitary channels may be observed.
Of course, with respect to the general description of channels above,
mixed-unitary channels are channels for which $m = n$, and it is evident
that every mixed-unitary channel $\Phi:\M_n\rightarrow\M_n$ is
\emph{unital}, meaning that $\Phi(\I_n) = \I_n$.
It is known that in the case $n=2$, a channel is mixed unitary if and only if
it is unital, but when $n\geq 3$ there exist channels
$\Phi:\M_n\rightarrow\M_n$ that are unital but not mixed unitary
\cite{Tregub1986,Landau1993}.

The importance of mixed-unitary channels in quantum information theory
is multifarious.
Many natural examples of noisy quantum channels (including the so-called
\emph{dephasing} and \emph{depolarizing} channels) are mixed unitary.
The general form of a mixed-unitary channel---which can be described as a
classical source of randomness selecting a unitary transformation to be applied
to a system---is simple and intuitive, and arises naturally in algorithmic and
cryptographic settings.
For example, the most standard forms of encryption and decryption of quantum
information using a private key induce mixed-unitary channels on the data
from the viewpoint of an eavesdropper \cite{AmbainisMTdW2000,HaydenLSW2004}.
\emph{Quantum expanders}, \emph{twirling operations,} and 
\emph{unitary $t$-designs} are specific types of
mixed-unitary channels that have been studied extensively in quantum
information and computation \cite{Ben-AroyaST2008,BDSW96,DLT01,DankertCEL2009}.
Mixed-unitary channels are also known to correspond precisely to those channels
for which an ideal form of environment-assisted error correction is possible
\cite{GregorattiW2003}.
Despite the fact that mixed-unitary channels have a simple form, they do
nevertheless inherit many interesting properties of general channels
\cite{Rosgen2008}.

Mixed-unitary channels also have important uses, as an analytic tool, in
quantum information theory.
For example, \emph{majorization} for Hermitian matrices \cite{AlbertiU1982} is
typically defined in terms of mixed-unitary channels: a Hermitian matrix $H$
is said to majorize a Hermitian matrix $K$ if there exists a mixed-unitary
channel $\Phi$ such that $\Phi(H) = K$.
This notion has found many applications in quantum information, perhaps most
notably in Nielsen's theorem \cite{Nielsen1999}, which provides a perfect
characterization of the bipartite pure state transformations that can be
realized through local quantum operations and classical communication.
Another example is that the monotonicity of quantum relative entropy under the
action of mixed-unitary channels, which follows directly from the joint
convexity of quantum relative entropy, offers a convenient stepping stone to
monotonicity for all channels.
(Recent proofs of the monotonicity of quantum relative entropy under the
action of all channels, and indeed all positive and trace-preserving maps
\cite{Mueller-HermesR2017}, do however offer an alternative path.)

Mixed-unitary channels are also interesting mathematical objects in their own
right, and have inspired fruitful lines of research.
For example, the \emph{asymptotic quantum Birkhoff conjecture}
\cite{SmolinVW2005}, which was eventually refuted \cite{HaagerupM2011}, was
concerned with the approximation of tensor powers of unital channels by
mixed-unitary channels.
Through the \emph{channel--state correspondence}, which essentially identifies a
channel with the state obtained by normalizing its Choi representation,
mixed-unitary channels also offer an interesting twist on bipartite
separability.
That is, whereas a separable state is a convex mixture of
\emph{pure product states}, the states corresponding to mixed-unitary channels
are convex mixtures of \emph{maximally entangled states}.
As it turns out, the two sets of states share some important common properties,
including the fact that they have a nonempty interior
\cite{ZyczkowskiHSL1998, GurvitsB2002, Watrous2009} and
have NP-hard membership testing problems
\cite{Gurvits2003, Ioannou2007, Gharibian2010, LeeW2019}.
Additional properties of mixed-unitary channels can be found in
\cite{Audenaert2008a} and \cite{MendlW2009}.

In this paper we focus on the minimum number $N$, over all possible expressions
of the form~\eqref{eq:general mixed-unitary channel} that exist for a given
mixed-unitary channel $\Phi$, and we shall refer to this number as the
\emph{mixed-unitary rank} of $\Phi$.
It is immediate that $N\geq r$; the mixed-unitary rank is always at least the
Choi rank.
Buscemi \cite{Buscemi2006} proved the upper bound $N \leq r^2$, which was the
strongest bound known prior to our work.\footnote{%
  See also Theorem 4.10 and Corollary 4.11 in \cite{Watrous2018}, which provides
  an alternative way to prove Buscemi's bound.
  We note, in addition, that one can obtain a very slight improvement to
  Buscemi's bound in the specific case $r=n^2$ by observing that the bound
  $N\leq n^4 - 2n^2 + 2$ follows in a straightforward fashion from
  Carath\'eodory's theorem, as is explained in Proposition 4.9 of
  \cite{Watrous2018}.}
With respect to the lower bound $N \geq r$, to our knowledge no examples of
mixed-unitary channels for which $N > r$ have previously been exhibited.
The main contributions of our paper are as follows.
\begin{enumerate}
\item
  We prove that the mixed-unitary rank $N$ of every mixed-unitary channel
  $\Phi:\M_n\rightarrow\M_n$ having Choi rank $r$ satisfies
  \begin{equation}
    N \leq r^2 - s + 1,
  \end{equation}
  where $s = \op{dim}(\mathcal{S}_{\Phi})$ is the dimension of the operator system
  associated with $\Phi$.
  This operator system is given by
  \begin{equation}
    \mathcal{S}_{\Phi} = \op{span}\{A_j^{\ast} A_k\,:\,1\leq j,k\leq r\},
  \end{equation}
  for any expression of $\Phi$ taking the form \eqref{eq:Kraus form}.
  By examining relations between the dimension of the operator
  system of a mixed-unitary channel and its Choi rank, we conclude that
  \begin{equation}
   N \leq r^2-r+1
  \end{equation}
  for all mixed-unitary channels.
  Furthermore, we prove that $N=r$ if either $r\leq 2$ or $s\leq 3$, and that
  $N\leq 6$ in the case when $r=3$.

\item
  We provide a construction through which one may obtain examples of
  mixed-unitary channels having mixed-unitary ranks strictly larger than their
  Choi ranks.
  Specifically, the construction takes these channels to be the
  \emph{direct sum} of a unitary channel with a mixed-unitary channel that can
  be expressed uniquely as a nontrivial convex combination of unitary channels.
  Through this construction we exhibit examples of mixed-unitary channels
  of the form $\Phi:\M_{d+1}\rightarrow\M_{d+1}$ having Choi rank $d+1$ and
  mixed-unitary rank $2d$ for every odd prime~$d$, as well as mixed-unitary
  Schur channels of the form $\Phi:\M_{d^2+1}\rightarrow\M_{d^2+1}$
  having Choi rank $d+1$ and mixed-unitary rank $2d$ for every positive integer
  $d$ for which $d+1$ mutually unbiased bases for the space $\complex^d$ exist
  (which includes every prime power $d$).

\item
  We observe that the mixed-unitary rank is not multiplicative with respect to
  tensor products.
  In particular, there exist mixed-unitary channels $\Phi$ and $\Psi$ having
  mixed-unitary ranks~4 and~2, respectively, such that the mixed-unitary rank
  of $\Phi\otimes\Psi$ is 6.

\item
  Finally, we examine the mixed-unitary ranks of the Werner--Holevo channels,
  which are an important class of channels in quantum information theory
  defined for every dimension $n\geq 2$ by the formulas
  \begin{equation}
    \Phi_0(X) = \frac{\Tr(X)\I_n + X}{n+1}
    \qquad\text{and}\qquad
    \Phi_1(X) = \frac{\Tr(X)\I_n - X}{n-1}
  \end{equation}
  for all $X\in\M_n$.
  The channel $\Phi_0$ is mixed unitary for all such $n$, while $\Phi_1$ is
  mixed-unitary if and only if $n$ is even.
  We prove the following facts concerning the mixed-unitary ranks of these
  channels.
  \begin{itemize}
  \item
    For all even $n \geq 2$, the Choi rank and the mixed-unitary rank are in
    agreement for both $\Phi_0$ and $\Phi_1$:
    we have $N = r = \binom{n+1}{2}$ for $\Phi_0$ and
    $N = r = \binom{n}{2}$ for $\Phi_1$.

  \item
    The mixed-unitary rank of $\Phi_0$ in every odd dimension $n$ is at most
    $n(n+3)/2$, and the mixed-unitary rank of $\Phi_0$ for the case $n=3$ is
    $N = r = 6$.
  \end{itemize}

  Numerical evidence suggests that the mixed-unitary rank of $\Phi_0$
  for a few other small odd values of~$n$ satisfies $N = r = \binom{n+1}{2}$,
  but we leave open the problem of determining whether or not this formula
  holds in general.
 
\end{enumerate}

\section{Preliminaries}

In this section we summarize known facts and results concerning quantum
channels, including a few results specific to mixed-unitary channels, that will
be used later in the paper.
Further information on quantum channels, and the role they play in the theory
of quantum information, can be found in texts on the subject, including
\cite{NielsenC2000,Wilde2017,Watrous2018}.

\subsection{Linear algebra notations and conventions}

Given any matrix $A\in\M_{m,n}$, we denote by $A^\t$, $\overline{A}$, and
$A^{\ast}$ the transpose, entry-wise conjugate, and adjoint (or conjugate
transpose) of $A$, respectively.
A square matrix $H\in\M_n$ is \emph{Hermitian} if $H = H^{\ast}$, a square
matrix $U\in\M_n$ is \emph{unitary} if $U^{\ast} U = U U^{\ast} = \I_n$, and
a matrix $A\in\M_{m,n}$ is an \emph{isometry} if $A^*A = \I_n$.
This last condition requires that $m\geq n$, and in the case $m=n$ the
condition that $A$ is an isometry is equivalent to $A$ being unitary.

The vectorization mapping $\vec: \M_{m,n}\rightarrow \complex^{nm}$ converts a
given matrix to a column vector by transposing its rows into columns and
stacking them on top of one another from top to bottom.
In more precise terms, this mapping is defined as
\begin{equation}
\vec(A) = \sum_{j=1}^m \sum_{k = 1}^n A(j,k) \, e_j \otimes e_k
\end{equation}
for every $A\in\M_{m,n}$, where $e_j\in\complex^m$ and $e_k\in\complex^n$
denote the elementary unit vectors having a 1 in entry $j$ or $k$,
respectively, and 0 in all other entries.
We define the inner product of two matrices $A, B \in \M_{m,n}$ as
$\ip{A}{B}  = \Tr(A^*B)$, which is equivalent to the ordinary inner
product (conjugate linear in the first argument) of $A$ and $B$ viewed as
vectors: $\ip{A}{B} = \ip{\vec(A)}{\vec(B)}$.

Finally, the adjoint of a linear map $\Phi:\M_n\rightarrow\M_m$ is the unique
linear map $\Phi^{\ast}:\M_m\rightarrow\M_n$ that satisfies
$\ip{Y}{\Phi(X)} = \ip{\Phi^{\ast}(Y)}{X}$ for all $X\in\M_n$ and $Y\in\M_m$.
The condition that $\Phi$ is \emph{trace-preserving} is equivalent to
$\Phi^*$ being unital.

\subsection{Choi and Kraus representations of channels}

We have already defined the Choi representation and the notion of a Kraus
representation of a completely positive linear map $\Phi:\M_n\rightarrow\M_m$
in the introduction, but it will be helpful to note two additional facts
concerning them.
First, if $\Phi$ is a completely positive map having the Kraus representation
\eqref{eq:Kraus form}, then its Choi representation is given by
\begin{equation}
  J(\Phi) = \sum_{k=1}^r \vec(A_k)\vec(A_k)^*.
\end{equation}
Second, although Kraus representations are not unique, any two Kraus
representations of a given completely positive map are related in the following
way: if $\Phi:\M_n\rightarrow\M_m$ has the Kraus representation
\eqref{eq:Kraus form}, and
\begin{equation}
  \Phi(X) = \sum_{k = 1}^N B_k X B_k^{\ast}
\end{equation}
is a Kraus representation of $\Phi$ for which $N \geq r$, then there must exist
an isometry $V\in\M_{N,r}$ such that
\begin{equation}
  B_k = \sum_{j=1}^r V(k,j) A_j
\label{eq:krausfreedom}
\end{equation}
for every $k\in\{1,\ldots,N\}$.

\subsection{Complementary channels}

Suppose that $\Phi:\M_n\rightarrow\M_m$ is a channel (i.e., a completely
positive and trace-preserving linear map) having Kraus representation
\begin{equation}
  \Phi(X) = \sum_{k = 1}^r A_k X A_k^{\ast}.
\end{equation}
The linear map $\Psi:\M_n\rightarrow\M_r$ defined by
\begin{equation}
  \Psi(X) = \sum_{j,k=1}^r \ip{A_k^{\ast} A_j}{X} \, E_{j,k}
\end{equation}
for all $X\in\M_n$ is then also necessarily a channel, and is said to be
\emph{complementary} to $\Phi$.
(The notion of a complementary channel is more commonly defined through the
Stinespring representation of channels, which we have no need to discuss, but
the definitions are equivalent.)

If it is the case that $r = \op{rank}(J(\Phi))$ and $\Psi:\M_n\rightarrow\M_r$ is
complementary to $\Phi$, then any other given channel
$\Xi:\M_n\rightarrow\M_N$ is also complementary to $\Phi$ if and only if
there exists an isometry $V \in \M_{N,r}$ such that
$\Xi(X) = V\Psi(X)V^*$ for every $X\in\M_n$.

\subsection{The operator system of a channel}

Let $n$ be a positive integer.
A linear subspace $\mathcal{S}\subseteq\M_n$ is an \emph{operator system} if
$\I_n\in\mathcal{S}$ and if $A^*\in\mathcal{S}$ for each
$A\in\mathcal{S}$.
Every operator system is spanned by its Hermitian elements.
In particular, if $s=\op{dim}(\mathcal{S})$ is the dimension of this subspace,
there exist Hermitian matrices $H_1,\dots,H_{s-1}$ such that
\begin{equation}
  \mathcal{S} = \op{span}\{\I_n,H_1,\dots,H_{s-1}\}.
\end{equation}
Further information on the topic of operator systems can be found in
\cite{Paulsen1986}, \cite{Conway1999}, and \cite{Paulsen2003}.

An operator system $\mathcal{S}_{\Phi}$ is associated with every channel
$\Phi:\M_n\rightarrow\M_m$ in the following manner:
\begin{equation}\label{eq:opsysphi}
  \mathcal{S}_\Phi = \{A\,:\, \vec(A)\in\op{im}(J(\Phi^*\Phi))\}.
\end{equation}
Equivalently, given any Kraus representation
\begin{equation}
  \Phi(X) = \sum_{k = 1}^r A_k X A_k^{\ast}
\end{equation}
of a channel $\Phi$, the operator system $\mathcal{S}_{\Phi}$ may be expressed as
\begin{equation}
  \mathcal{S}_{\Phi} = \op{span}\bigl\{A_k^{\ast} A_j\,:\,1\leq j,k\leq r\bigr\},
\end{equation}
which is evident from the observation that
\begin{equation}
  J(\Phi^*\Phi) = \sum_{j,k=1}^r \vec(A_k^*A_j) \vec(A_k^*A_j)^*.
\end{equation}
The fact that $\mathcal{S}_{\Phi}$ is closed under adjoints is immediate, while
the condition $\I_n \in \mathcal{S}_{\Phi}$ follows from the assumption that $\Phi$
preserves trace, and therefore satisfies
\begin{equation}
  \sum_{k = 1}^r A_k^{\ast} A_k = \I_n.
\end{equation}
The operator system $\mathcal{S}_{\Phi}$ of the channel $\Phi$ has also been referred to
as the \emph{non-commutative graph} of $\Phi$ in the context of quantum
zero-error information theory \cite{Duan2013}.

If $\Phi:\M_n\rightarrow\M_m$ is a channel with Choi rank $r$, then
the dimension $s = \op{dim}(S_{\Phi})$ of the operator system of
$\Phi$ necessarily satisfies $s \leq r^2$.  One also has that $s =
r^2$ if and only if $\Phi$ is an extreme point in the convex set of
all channels from $\M_n$ to $\M_m$ \cite{Choi1975}. In addition, if $\Phi$ is a mixed-unitary channel with mixed-unitary rank $N$, it further holds that $s \le
  N^2 -N +1$, as each unitary operator $U_k$ in the Kraus representation \eqref{eq:Kraus form} satisfies $U_k^*U_k = \I_n$.

\subsection{Direct sums of channels}

Let $n$ and $m$ be positive integers and let $\Phi:\M_n\rightarrow\M_n$ and
$\Psi:\M_m\rightarrow\M_m$ be linear mappings.
The \emph{direct sum} of the mappings $\Phi$ and $\Psi$ is the linear map
$\Phi\oplus\Psi:\M_{n+m} \rightarrow \M_{n+m}$ defined as
\begin{equation}
(\Phi\oplus\Psi) \begin{pmatrix}
  X & \cdot \\
  \cdot& Y
\end{pmatrix} = \begin{pmatrix}
  \Phi(X) & 0\\
  0& \Psi(Y)
\end{pmatrix}
\end{equation}
for every $X\in\M_n$ and $Y\in\M_m$ (where the dots indicate arbitrary matrices
of the appropriate size that have no influence on the output of the map).
If $\Phi$ and $\Psi$ are channels, then $\Phi\oplus\Psi$ is a channel as well.
It is always the case that
\begin{equation}
  \op{rank}(J(\Phi\oplus\Psi)) = \op{rank}(J(\Phi)) + \op{rank}(J(\Psi)) .
\end{equation}
Using any Kraus representations of $\Phi$ and $\Psi$, 
we can express $\Phi\oplus\Psi$ as 
\begin{equation}
  (\Phi\oplus\Psi) \begin{pmatrix} X & \cdot \\
  \cdot& Y \end{pmatrix} = \sum_{k=1}^N\begin{pmatrix} A_k & 0 \\0& B_k \end{pmatrix}
  \begin{pmatrix} X & \cdot \\\cdot& Y \end{pmatrix}
  \begin{pmatrix} A_k & 0 \\0& B_k \end{pmatrix}^*
\end{equation}
for some choice of matrices $A_1,\dots,A_N\in\M_n$ and $B_1,\dots,B_N\in\M_m$ 
and for some positive integer $N$. 
Furthermore, following the relation between different Kraus representations 
given by (\ref{eq:krausfreedom}) 
every Kraus representation of $\Phi\oplus\Psi$ must have the form similar
to the above.

\subsection{Unique mixed-unitary decompositions}

Let $n$ be a positive integer, let $\Phi:\M_n\rightarrow\M_n$ be a
mixed-unitary channel, and let $N$ be the mixed-unitary rank of $\Phi$.
Let us also introduce the following notation: for two unitary matrices
$U,V\in\U_n$, we write $U\sim V$ if there exists a complex number
$\alpha\in\complex$ with $\abs{\alpha} = 1$ such that $U = \alpha V$.
That is, $U\sim V$ if and only if $U$ and $V$ describe the same unitary
channel through conjugation.
A decomposition
\begin{equation}
  \Phi(X) = \sum_{k=1}^N p_k U_k X U_k^{\ast}
\end{equation}
of $\Phi$ is said to be a \emph{unique mixed-unitary decomposition} for $\Phi$
if the following statement is true.
For every mixed-unitary decomposition
\begin{equation}
  \Phi(X) = \sum_{j = 1}^t q_j V_j X V_j^{\ast}
\end{equation}
of $\Phi$, there must exist a partition
$\{ 1,\ldots,t \} = T_1 \cup \cdots \cup T_N$ so that these two conditions hold
for every $k\in\{1,\ldots,N\}$:
\begin{enumerate}
  \item
    $V_j \sim U_k$ for every $j\in T_k$.
  \item $p_k = \sum_{j\in T_k}q_j$.
\end{enumerate}

\subsection{Condition for a channel to be mixed unitary}

Finally, we require a characterization of mixed-unitary channels,
expressed by the following theorem and corollary.
The theorem is based on a characterization of mixed-unitary channels due to
Audenaert and Scheel~\cite{Audenaert2008a}; a proof is included below because
we require a slight refinement of their characterization.

A square matrix $X\in\M_n$ is said to be \emph{traceless} if $\Tr(X) =0$ and is said to have \emph{vanishing diagonal} if all of its diagonal entries are equal to $0$. We remark that the set of traceless $n\times n$ matrices is equal to the set $\I_n^\perp$.

\begin{theorem}
  \label{thm:condition}
  Let $\Phi:\M_n\rightarrow\M_n$ be a channel having Choi rank $r$.
  For every positive integer $N\geq r$, the following statements are
  equivalent:
  \begin{enumerate}
  \item $\Phi$ is mixed unitary with mixed-unitary rank at most $N$.
  \item There is a channel $\Psi:\M_n\rightarrow\M_N$ complementary
    to $\Phi$ such that $\Psi(X)$ has vanishing diagonal for every traceless
    matrix $X$.
 \end{enumerate}
\end{theorem}

\begin{proof}
  Suppose first that $\Psi:\M_n\rightarrow\M_N$ is a channel complementary to
  $\Phi$ such that $\Psi(X)$ has vanishing diagonal for every traceless
  matrix $X$.
  Let $A_1,\dots,A_N\in\M_n$ be matrices satisfying
  \begin{equation}
    \Phi(X) = \sum_{k=1}^N A_kXA_k^*
    \qquad\text{and}\qquad\ip{A_k^*A_j}{X} E_{j,k}
  \end{equation}
  for each $X\in\M_n$.
For each index $k\in\{1,\dots,N\}$, one has that
    \begin{equation}
    \ip{A_k^*A_k}{X}  = \ip{E_{k,k}}{\Psi(X)}
    = 0
  \end{equation}
 for each matrix $X\in\I_n^\perp$, as the diagonal entries of $\Psi(X)$ are equal to zero by assumption for each traceless matrix $X$. It must therefore hold that $A_k^*A_k\in\op{span}\{\I_n\}$ and thus
  $A_k = \alpha_k U_k$ for some choice of a complex number $\alpha_k$ and
  a unitary matrix $U_k\in\U_n$, for every index $k\in\{1,\dots,N\}$.
  It follows that 
  \begin{equation}
    \Phi(X) = \sum_{k=1}^N A_k X A_k^*
    = \sum_{k=1}^N \abs{\alpha_k}^2 U_k X U_k^*
    = \sum_{k=1}^N  p_k U_k X U_k^*
  \end{equation}
  for each $X\in\M_n$, where $(p_1,\dots,p_N)$ is the probability vector
  defined as $p_k=\abs{\alpha_k}^2$ for each $k\in\{1,\dots,N\}$, and thus
  $\Phi$ is mixed unitary with mixed-unitary rank at most $N$.
  
  To prove the reverse implication, suppose that there exist unitary matrices
  $U_1,\dots,U_N\in\U_n$ and a probability vector $(p_1,\dots,p_N)$
  satisfying
  \begin{equation}
    \Phi(X) = \sum_{k=1}^N p_k U_k X U_k^*
  \end{equation}
  for each $X\in\M_n$, and define a channel $\Psi:\M_n\rightarrow\M_N$
  complementary to $\Phi$ as
  \begin{equation}
    \Psi(X) = \sum_{j,k=1}^N \sqrt{p_jp_k}\ip{U_k^*U_j}{X} E_{j,k} 
  \end{equation}
  for each $X\in\M_n$.
  For each traceless matrix $X$, one has
  \begin{equation}
    \ip{E_{k,k}}{\Psi(X)}  = p_k
    \ip{U_k^*U_k}{X} = p_k \Tr(X) = 0
  \end{equation}
  for each $k\in\{1,\dots,N\}$, and therefore the diagonal entries of $\Psi(X)$
  are equal to zero.
\end{proof}

\begin{corollary} \label{cor:method}
  Let $n$ be a positive integer, let $\Phi:\M_n\rightarrow\M_n$ be a channel,
  let $r$ be the Choi rank of $\Phi$, and let $\Psi:\M_n\rightarrow\M_r$ be a
  channel that is complementary to $\Phi$.
  For every integer $N\geq r$, the channel $\Phi$ is mixed unitary with
  mixed-unitary rank at most $N$ if and only if there exists an isometry
  $V\in\M_{N,r}$ such that $V\Psi(X)V^*$ has vanishing diagonal for each
  traceless matrix $X$.
\end{corollary}

\begin{proof}
 This follows from Theorem \ref{thm:condition} and the observation that another channel $\Xi:\M_n\rightarrow\M_N$ is complementary to
$\Phi:\M_n\rightarrow\M_n$ if and only if there exists an isometry
$V\in\M_{N,r}$ satisfying $\Xi(X)= V\Psi(X)V^*$ for each $X\in\M_n$.
\end{proof}

Corollary \ref{cor:method} suggests a useful method for approximating the
mixed-unitary rank of a channel.
Given a channel $\Phi:\M_n\rightarrow\M_n$ with Choi rank $r$ and a candidate
integer $N\geq r$ for the mixed-unitary rank of $\Phi$, one can use the
following procedure to determine if $\Phi$ can be decomposed as a convex
combination of $N$ unitary channels:
 \begin{enumerate}
 \item Choose any channel $\Psi:\M_n\rightarrow\M_r$ that is complementary
   to $\Phi$.
 \item Define a subspace of matrices $\mathcal{B} \subseteq \M_r$ as
   $\mathcal{B} = \{\Psi(X)\, :\, X\in\I_n^\perp\}$ and choose matrices
   $B_1,\dots,B_N\in\mathcal{B}$ such that
   $\mathcal{B} = \op{span}\{B_1,\dots,B_N\}$. 
 \item Use numerical methods to search for an isometry
   $V\in\M_{N,r}$ such that $\ip{E_{j,j}}{VB_kV^*} = 0 $ for each pair
   of indices $j\in\{1,\dots,N\}$ and $k\in\{1,\dots,r\}$. 
 \end{enumerate}
 The last step cannot be performed efficiently (unless
 $\mathrm{P}=\mathrm{NP}$), but we have found that this procedure is useful for
 computing the mixed-unitary ranks of interesting channels for small choices of
 the dimension~$n$.

\section{Upper bounds on mixed-unitary rank}
\label{sec:mixedunitaryrank}

In this section we prove upper bounds on the mixed-unitary rank of
mixed-unitary channels.
We begin with the following general theorem.

\begin{theorem}
  \label{thm:upperbound}
  Let $n$ be a positive integer, let $\Phi:\M_n\rightarrow\M_n$ be a
  mixed-unitary channel having Choi rank~$r$ and mixed-unitary rank $N$, and
  let $s = \op{dim}(\mathcal{S}_{\Phi})$ be the dimension of the operator system
  of $\Phi$.
  It is the case that
 \begin{equation} \label{eq:rMr2d1}
   N \leq r^2 - s + 1.
 \end{equation}
\end{theorem}

\begin{proof}
  Let $A_1,\dots,A_r\in\M_n$ be matrices offering a Kraus representation
  of $\Phi$:
  \begin{equation}
    \Phi(X)= \sum_{k=1}^r A_k X A_k^*
  \end{equation}
  for all $X\in\M_n$.
  Also define $\Psi:\M_n\rightarrow\M_r$ as
  \begin{equation}
    \Psi(X) = \sum_{j,k=1}^r \ip{A_k^*A_j}{X} E_{j,k}
  \end{equation}
  for all $X\in\M_n$, so that $\Psi$ is a channel complementary to $\Phi$.
  It is the case that $\op{im}(\Psi^*) = \mathcal{S}_\Phi$ and therefore
  $s = \op{dim}(\op{im}(\Psi^*))$.
  By the rank-nullity theorem we have that
  \begin{equation}
    \op{dim}(\ker(\Psi^*)) + \op{dim}(\op{im}(\Psi^*)) = \op{dim}(\M_r) = r^2,
  \end{equation}
  and so the theorem will follow from a demonstration that
  $N\leq \op{dim}(\ker(\Psi^*))+1$.  
  Toward this goal, suppose $U_1,\dots,U_N\in\U_n$ are unitary matrices
  and $(p_1,\dots,p_N)$ is a probability vector such that
  \begin{equation}
    \Phi(X)=\sum_{k=1}^N p_k U_k X U_k^*
  \end{equation}
  for each $X\in\M_n$, and let us observe that each $p_k$ must be nonzero by
  the assumption that $N$ is the mixed-unitary rank of $\Phi$.

  At this point we have two Kraus representations of $\Phi$, which must be
  related as was discussed in the previous section:
  there must exist an isometry $V\in\M_{N,r}$ such that 
  \begin{equation}
    \sqrt{p_k}U_k = \sum_{j=1}^r V(k,j) A_j 
  \end{equation}
  for every $k\in\{1,\dots,N\}$.
Define $u_k \in\complex^r$ as $u_k=V^{\t} e_k$ for each $k\in\{1,\dots,N\}$ and define a matrix  $A\in\M_{n^2,r}$ as
  \begin{equation}
    A = \sum_{j=1}^r \vec(A_j)e_j^*.
  \end{equation}
  Observe that
  \begin{equation}
    \sqrt{p_k}\vec(U_k) = A V^{\t} e_k = A u_k
  \end{equation}
  for each $k\in\{1,\dots,N\}$.

  Next, consider the collection
  \begin{equation}
    \label{eq:vecUkvecUk}
    \{\vec(U_1)\vec(U_1)^*,\,\dots,\,\vec(U_N)\vec(U_N)^*\},
  \end{equation}
  which we claim must be linearly independent.
  To verify this claim, suppose to the contrary that
  $\alpha_1,\ldots,\alpha_N\in\complex$ are not all zero and satisfy
  \begin{equation}
    \alpha_1 \vec(U_1)\vec(U_1)^* + \cdots + \alpha_N \vec(U_N)\vec(U_N)^* = 0.
  \end{equation}
  By taking the trace of both sides of this equation, we find that
  $\alpha_1 + \cdots + \alpha_N = 0$, and therefore
  \begin{equation}
    \sum_{k = 2}^N
    \alpha_k \bigl( \vec(U_k)\vec(U_k)^{\ast} - \vec(U_1)\vec(U_1)^{\ast}\bigr)
    = 0.
  \end{equation}
  As $\alpha_1 + \cdots + \alpha_N = 0$, it cannot be that
  $\alpha_2,\ldots,\alpha_N$ are all zero, so the collection
  \begin{equation}
    \bigl\{
    \vec(U_2)\vec(U_2)^{\ast} - \vec(U_1)\vec(U_1)^{\ast},\ldots,
    \vec(U_N)\vec(U_N)^{\ast} - \vec(U_1)\vec(U_1)^{\ast}
    \bigr\}
  \end{equation}
  is linearly dependent, implying that the collection \eqref{eq:vecUkvecUk}
  generates an affine subspace of dimension strictly less than $N-1$.
  This, however, contradicts the assumption that the mixed-unitary rank of
  $\Phi$ is $N$ through Carath\'eodory's theorem.
  
  Given that the collection \eqref{eq:vecUkvecUk} is linearly independent and
  each $p_k$ is nonzero, it follows that the collection
  \begin{equation}
    \{Au_ku_k^*A^*,\dots,Au_N u_N^*A^*\}
    = \{p_1\vec(U_1)\vec(U_1)^*,\dots,p_N \vec(U_N)\vec(U_N)^*\}
  \end{equation}
  is also linearly independent.
  This implies that the collection $\mathcal{B}\subseteq\M_r$ defined as 
  \begin{equation}
    \mathcal{B}= \{u_1u_1^*,\dots,u_N u_N^*\}
  \end{equation}
  is linearly independent as well: $\op{dim}(\op{span}(\mathcal{B}))=N$.
  For each index $k\in\{1,\dots,N\}$ we see that
  \begin{equation}
    \Psi^*(u_ku_k^*)
    = \sum_{i,j=1}^r \bigip{E_{i,j}}{V^{\t}E_{k,k}\overline{V}}
    A_j^*A_i
    = \sum_{i,j=1}^r V(k,i) \overline{V(k,j)} A_j^*A_i
    = p_k U_k^* U_k = p_k\I_n,
  \end{equation}
  and therefore $\Psi^*(\op{span}(\mathcal{B}))=\op{span}\{\I_n\}$.
  It follows that $\op{dim}(\ker(\Psi^*)\cap\op{span}(\mathcal{B})) = N-1$ which
  implies that $\op{dim}(\ker(\Psi^*))\geq N-1$, completing the proof.
\end{proof}

\begin{corollary}
 Let $n$ be a positive integer, let $\Phi:\M_n\rightarrow\M_n$ be a
 mixed-unitary channel with mixed-unitary rank equal to $N$ and Choi rank equal
 to $r$.
 It is the case that
 \begin{equation}
   N\leq r^2-r+1.
 \end{equation}
\end{corollary}

\begin{proof}
  By Theorem~\ref{thm:upperbound}, it will suffice to show that
  $r\leq \op{dim}(\mathcal{S}_\Phi)$.
  Let $U_1,\dots,U_N\in\U_n$ be unitary matrices and let $(p_1,\dots,p_N)$ be a
  probability vector for which
  \begin{equation}
    \Phi(X)=\sum_{k=1}^Np_kU_kXU_k^*
  \end{equation}
  for each $X\in\M_n$, and consider the collection of matrices
  \{$U_1^*U_1,\dots,U_1^*U_N\}\subset\mathcal{S}_\Phi$.
  As $U_1$ is invertible, one has that
  \begin{equation}
    \op{dim}(\{U_1^*U_1,\dots,U_1^*U_N\})
    = \op{dim}(\{U_1,\dots,U_N\}) = \op{rank}(J(\Phi))=r
  \end{equation}
  from which it follows that $\op{dim}(\mathcal{S}_\Phi)\geq r$.
\end{proof}

Next we will prove that any channel $\Phi:\M_n\rightarrow\M_n$ having an
operator system of dimension 3 or less must be mixed unitary, and indeed must
have mixed-unitary rank in agreement with its Choi rank.
By combining this theorem with the previous one, we obtain the upper bound
$\max\{r,r^2-3\}$ on the mixed-unitary rank of any mixed-unitary channel having
Choi rank $r$.
In the proof of the theorem to follow, we will use the fact that every
traceless square matrix is unitarily equivalent to one having a vanishing
diagonal.
This is a well-known fact that follows from the Toeplitz--Hausdorff theorem.
(See, for instance, Theorem 1.3.4 of \cite{HornJ1994}.)

\begin{theorem}\label{thm:lowdimension}
  Let $n$ be a positive integer, let $\Phi:\M_n\rightarrow\M_n$ be a channel,
  and suppose that $\op{dim}(\mathcal{S}_\Phi)\leq 3$.
  The channel $\Phi$ is mixed unitary with mixed-unitary rank equal to its Choi
  rank.
\end{theorem}

\begin{proof}
 Let $r=\op{rank}(J(\Phi))$ be the Choi rank of $\Phi$, let
 $\Psi:\M_n\rightarrow\M_r$ be any channel complementary to $\Phi$,
 and observe that $\mathcal{S}_\Phi = \op{im}(\Psi^*)$. Define a subspace of matrices $\mathcal{A}\subseteq\M_r$ as
 \begin{equation}
   \mathcal{A} = \{\Psi(X)\,:\, X\in\I_n^\perp\}.
 \end{equation}
 As $\mathcal{A}$ is the image of all traceless matrices under the action of $\Psi$, by Corollary \ref{cor:method} it will suffice to show the existence of a unitary matrix $U\in\U_r$ such that $U\mathcal{A}U^*$ contains only matrices with vanishing diagonal. By the assumption $\op{dim}(\mathcal{S}_\Phi)\leq 3$, together with the
 observation that $\I_n\in \mathcal{S}_{\Phi}$, we conclude that there must exist
 traceless Hermitian matrices $H,K\in\I_n^\perp$ such that
 \begin{equation}
   \mathcal{S}_\Phi = \op{span}\{\I_n,H,K\}.
 \end{equation}
 As $\op{im}(\Psi^*)=\ker(\Psi)^\perp$, one therefore has that $\op{im}(\Psi)=\Psi(\mathcal{S}_\Phi)$ from which we may conclude that $\mathcal{A} = \op{span}\{\Psi(H),\Psi(K)\}$.
 Note that $\Tr(A)=0$ for each $A\in\mathcal{A}$, as $\Psi$ is trace preserving,
 and in particular
 \begin{equation}
   \Tr(\Psi(H+iK))=0.
 \end{equation}
 There must therefore exist a unitary matrix $U\in\U_n$ such that
 $U\Psi(H+iK)U^*$ has vanishing diagonal.
 Because both $H$ and $K$ are Hermitian, each of the Hermitian matrices
 $U\Psi(H)U^*$ and $U\Psi(H)U^*$ must therefore also have vanishing diagonal.
 It follows that $UAU^*$ has vanishing diagonal for every $A\in\mathcal{A}$,
 and therefore $\Phi$ is mixed unitary with mixed-unitary rank equal to $r$ by
 Corollary~\ref{cor:method}.
\end{proof}

We may also use Theorem \ref{thm:lowdimension} to show that every channel with
Choi rank equal to 2 is either extremal or a convex combination of two unitary
channels.

\begin{corollary}\label{cor:rank2}
  Let $n$ be a positive integer and let $\Phi:\M_n\rightarrow\M_n$ be a
  channel with Choi rank equal to $2$.
  If $\Phi$ is not an extreme point in the set of all channels, then $\Phi$ is
  mixed unitary with mixed-unitary rank equal to 2.
\end{corollary}

\begin{proof}
  The channel $\Phi$ is an extreme point in the convex set of all channels
  if and only if $\op{dim}(\mathcal{S}_\Phi) = \op{rank}(J(\Phi))^2$.
  Therefore, under the assumption that $\Phi$ is not extremal, it 
  follows that $\op{dim}(\mathcal{S}_\Phi)\leq 3$ by the assumption that $\Phi$
  has Choi rank 2.
  The channel $\Phi$ is therefore mixed unitary with mixed-unitary rank
  equal to 2 by Theorem \ref{thm:lowdimension}.
\end{proof}

Finally, we may combine the results of Theorem \ref{thm:upperbound} and Theorem \ref{thm:lowdimension} to improve the upper bound on the mixed-unitary rank of mixed-unitary channels with Choi rank equal to 3.

\begin{corollary}
 Let $n$ be a positive integer, let $\Phi:\M_n\rightarrow\M_n$ be a
 mixed-unitary channel with Choi rank equal to 3 and mixed-unitary rank equal to $N$. It is the case that $N\leq 6$. 
\end{corollary}

\begin{proof}
   If $s = \op{dim}(\mathcal{S}_\Phi)\leq 3$, then $N=3$ by
  Theorem \ref{thm:lowdimension}; if $s\geq 4$, then $N\leq 6$ by
 Theorem~\ref{thm:upperbound}.
\end{proof}

\section{A construction for non-trivial mixed-unitary rank}

Next we present a construction to obtain mixed-unitary channels with
mixed-unitary ranks strictly larger than their Choi ranks.
The construction makes use of two concepts that were discussed in the
preliminaries section: the direct-sum of two channels and the notion of
a unique mixed-unitary decomposition of a mixed-unitary channel.

\begin{theorem}\label{thm:r2r}
  Let $n$ and $m$ be positive integers, let $\Phi:\M_n\rightarrow\M_n$ be a
  mixed-unitary channel having Choi rank equal to $r$, mixed-unitary rank
  equal to $r$, and a unique mixed-unitary decomposition.
  For every unitary channel $\Psi:\M_m\rightarrow\M_m$, the direct sum
  channel $\Phi\oplus\Psi$ has Choi rank $r+1$ and mixed-unitary rank $2r$.
\end{theorem}

\begin{proof}
  We begin by noting that there is no loss of generality in assuming $\Psi$ is
  the identity channel on $\M_m$, which we shall do for the remainder of the
  proof.
  The Choi rank of $\Phi\oplus\Psi$ can simply be calculated:
  \begin{equation}
    \op{rank}(J(\Phi\oplus\Psi)) = \op{rank}(J(\Phi)) + \op{rank}(J(\Psi)) = r+1.
  \end{equation}
  It therefore remains to prove that the mixed-unitary rank of $\Phi\oplus\Psi$
  is equal to $2r$.
  For the purpose of doing this, we let
  \begin{equation}
    \label{eq:unique decomposition}
    \Phi(X) = \sum_{k=1}^r p_k U_k X U_k^*
  \end{equation}
  be a unique mixed-unitary decomposition of $\Phi$, the existence of which has
  been assumed by the theorem.
  We also observe that
  \begin{equation}
    (\Phi\oplus \Psi)(Z)
    =\frac{1}{2} \sum_{k=1}^r p_k
    \begin{pmatrix}
      U_k & 0\\
      0 & \I_m
    \end{pmatrix}
    Z
    \begin{pmatrix}
      U_k & 0\\
      0 & \I_m
    \end{pmatrix}^* +
    \frac{1}{2}\sum_{k=1}^r p_k
    \begin{pmatrix}
      U_k & 0\\
      0 & -\I_m
    \end{pmatrix}
    Z
    \begin{pmatrix}
      U_k & 0\\
      0 & -\I_m
    \end{pmatrix}^*
  \end{equation}
  is a mixed-unitary decomposition of $\Phi\oplus\Psi$, establishing that its
  mixed-unitary rank is no greater than $2r$.

  To complete the proof, we must establish that every mixed-unitary
  decomposition of $\Phi\oplus\Psi$ has at least $2r$ terms.
  With this task in mind, and recalling from the discussion on direct
  sums of channels that each Kraus operator of $\Phi \oplus \Psi$
  is necessarily a direct sum of matrices,  
  consider any mixed-unitary decomposition
  \begin{equation}
    \label{eq:block mixed-unitary decomposition}
    (\Phi\oplus \Psi)(Z)
    = \sum_{j=1}^t q_j
    \begin{pmatrix}
      V_j & 0\\
      0 & W_j
    \end{pmatrix}
    Z
    \begin{pmatrix}
      V_j & 0\\
      0 & W_j
    \end{pmatrix}^*.
  \end{equation}
  We note that the decomposition \eqref{eq:block mixed-unitary decomposition}
  implies that
  \begin{equation}
    \Phi(X) = \sum_{j=1}^t q_j V_j X V_j^{\ast}
    \quad\text{and}\quad
    \Psi(Y) = \sum_{j=1}^t q_j W_j Y W_j^{\ast}
  \end{equation}
  for all $X\in\M_n$ and $Y\in\M_m$.
  Because $\Psi$ is the identity channel on $\M_m$, it must therefore be the
  case that $W_j\sim\I_m$ for all $j\in\{1,\ldots,t\}$.
  That is, there exist complex units $\alpha_1,\ldots,\alpha_t$ such that
  $W_j = \alpha_j\I_m$ for all $j\in\{1,\ldots,t\}$. 
  Moreover, by the assumption that \eqref{eq:unique decomposition}
  is a unique mixed-unitary decomposition, there must therefore exist a
  partition $\{1,\ldots,t\} = T_1 \cup\cdots\cup T_r$ such that
  for every $k\in\{1,\ldots,r\}$ we have $V_j \sim U_k$ for every
  $j\in T_k$ and $p_k = \sum_{j\in T_k}q_j$.
  Let us choose complex units $\beta_{j,k}$ for all $j\in T_k$ so that
  $V_j = \beta_{j,k} U_k$. 
  Finally, we observe that the decomposition
  \eqref{eq:block mixed-unitary decomposition}
  implies that
  \begin{equation}
    (\Phi\oplus \Psi)
    \begin{pmatrix}
      0 & Z\\
      0 & 0
    \end{pmatrix}
    = \sum_{j=1}^t q_j
    \begin{pmatrix}
      0 & V_j Z W_j^{\ast} \\
      0 & 0
    \end{pmatrix}
    = \sum_{j=1}^t q_j
    \begin{pmatrix}
      0 & \overline{\alpha_j}V_j Z \\
      0 & 0
    \end{pmatrix}
  \end{equation}
  for all $Z\in\M_{n,m}$.
  The direct sum of two channels must zero-out the off-diagonal blocks of its
  input, and therefore we conclude that
  \begin{equation}
    \sum_{j = 1}^t q_j \overline{\alpha_j} V_j = 0.
  \end{equation}
  By splitting this sum according to the partition $T_1\cup\cdots\cup T_r$, we
  find that
  \begin{equation}
    0 = \sum_{k = 1}^r \sum_{j\in T_k} q_j \overline{\alpha_j} V_j
    = \sum_{k = 1}^r\Biggl(
    \sum_{j\in T_k} q_j \overline{\alpha_j} \beta_{j,k}\Biggr) U_k.
  \end{equation}

  The matrices $U_1,\ldots,U_r$ are linearly independent by the assumption that
  $\Phi$ has Choi rank $r$.
  It is therefore the case that
  \begin{equation}
    \label{eq:partition sums 0}
    \sum_{j\in T_k} q_j \overline{\alpha_j} \beta_{j,k} = 0
  \end{equation}
  for every $k\in\{1,\ldots,r\}$.
  Note that it cannot be that there is any choice of $k$ such that
  $q_j = 0$ for every $j\in T_k$, for then we would have
  $p_k = \sum_{j\in T_k}q_j = 0$, which violates the assumption that
  $\Phi$ has Choi rank $r$.
  Each of the sums \eqref{eq:partition sums 0} is therefore a positive linear
  combination of complex units, and consequently $\abs{T_k} \geq 2$ for ever
  $k\in\{1,\ldots,r\}$.
  This implies $t = \abs{T_1} + \cdots + \abs{T_r} \geq 2r$, as required.
\end{proof}

Naturally, in order to use the previous theorem to obtain examples of
mixed-unitary channels whose mixed-unitary ranks are greater than their Choi
ranks, one must address the following question:
Under what conditions must a mixed-unitary channel $\Phi$ have a unique
mixed-unitary decomposition?
The following theorem provides one suitable condition:
if the dimension $s$ of the operator system of a mixed-unitary channel $\Phi$
satisfies $s = r^2 - r + 1$, for $r$ being the Choi rank of $\Phi$, then the
mixed-unitary rank of $\Phi$ must also be $r$, and moreover $\Phi$ possesses a
unique mixed-unitary decomposition. 

\begin{theorem}\label{thm:uniqueMU}
  Let $n$ be a positive integer and let $\Phi:\M_n\rightarrow\M_n$ be a
  mixed-unitary channel having Choi rank~$r$.
  If the dimension $s$ of the operator system of $\Phi$ is given by
  $s = r^2 - r + 1$, then $\Phi$ has mixed-unitary rank equal to $r$ and has a
  unique mixed-unitary decomposition.
\end{theorem}

\begin{proof}
  Suppose $\Phi$ has mixed-unitary rank equal to $N$.
  By Theorem \ref{thm:upperbound} we have $N \leq r^2 - s + 1 = r$, and
  therefore $N = r$.
  There must therefore exist distinct unitary matrices $U_1,\dots,U_r\in\U_n$ and a
  probability vector $(p_1,\dots,p_r)$ so that
  \begin{equation}
    \label{eq:unique mixed-unitary decomposition}
    \Phi(X) = \sum_{k=1}^r p_k U_k X U_k^*
  \end{equation}
  for every $X\in\M_n$.
  It remains to prove that \eqref{eq:unique mixed-unitary decomposition} is
  a unique mixed-unitary decomposition of $\Phi$.
  Toward this goal we define a map $\Psi:\M_n\rightarrow\M_r$ as
  \begin{equation}
    \Psi(X) = \sum_{j,k=1}^r \sqrt{p_jp_k}\, \ip{U_k^*U_j}{X}\, E_{j,k}
  \end{equation}
  for all $X\in\M_n$, so that $\Psi$ is a complementary channel to $\Phi$.
  For convenience we note that
  \begin{equation}
    \Psi^{\ast}(Y) = \sum_{j,k=1}^r \sqrt{p_j p_k}\,
    Y(j,k) \, U_k^*U_j
  \label{eq:comppsi}
  \end{equation}
  for all $Y\in\M_r$.

  We now observe that $\Psi^{\ast}(Y) \in \op{span}\{\I_n\}$ if and only if
  $Y\in\M_r$ is a diagonal matrix.
  Indeed, from (\ref{eq:comppsi}), for every diagonal matrix $Y$, 
  $\Psi^{\ast}(Y) \in\op{span}\{\I_n\}$. On the other hand, by the rank nullity theorem we have
  \begin{equation}
    \op{dim}(\ker(\Psi^{\ast})) = r^2 - \op{dim}(\op{im}(\Psi^{\ast})) = r^2 - s = r - 1,
  \end{equation}
  so the subspace containing all matrices $Y$ satisfying
  $\Psi^{\ast}(Y)\in\op{span}\{\I_n\}$ can have dimension no larger than $r$.
  Thus, there can be no matrices $Y$ outside of the $r$ dimensional subspace
  of diagonal matrices in $\M_r$ that satisfy
  $\Psi^{\ast}(Y) \in \op{span}\{\I_n\}$.
    
  Now suppose that $V_1,\dots,V_t\in\U_n$ are unitary matrices and
  $(q_1,\ldots,q_t)$ is a probability vector with each $q_k$ being positive
  such that
  \begin{equation}
    \Phi(X) = \sum_{k=1}^t q_k V_k X V_k^*
  \end{equation}
  for all $X\in\M_n$.
  It follows that there must exist an isometry $W\in\M_{t,r}$ satisfying
  \begin{equation}
    \sqrt{q_k} V_k = \sum_{j=1}^r W(k,j) \sqrt{p_j}\,U_j
  \end{equation}
  for each $k\in\{1,\dots,t\}$, from which we conclude that
  \begin{equation}
    \Psi^{\ast}\bigl(\overline{W} E_{k,k} W^{\t}\bigr)
    = \sum_{i,j=1}^r \sqrt{p_i p_j}\,
    W(k,i)\, \overline{W(k,j)} \,U_j^{\ast} U_i
    = q_k V_k^{\ast} V_k = q_k \I_n.
  \end{equation}
  It is therefore the case that $D_k = \overline{W} E_{k,k} W^{\t}$ is diagonal
  for every $k\in\{1,\ldots,t\}$.
  We conclude that
  \begin{equation}
    q_k V_k X V_k^{\ast}
    = \sum_{i,j=1}^r \sqrt{p_i p_j} D_k(j,i) U_i X U_j^{\ast}
    = \sum_{j=1}^r p_j D_k(j,j) U_j X U_j^{\ast}
  \end{equation}
  for every $X\in\M_n$ and $k\in\{1,\ldots,t\}$.
  As $U_1,\ldots,U_r$ are linearly independent, it follows that for every
  index $k\in\{1,\ldots,t\}$, the matrix entry $D_k(j,j)$ is nonzero for
  precisely one choice of an index $j\in\{1,\ldots,r\}$, and for this unique
  index $j$ it is necessarily the case that $V_k\sim U_j$.
  Letting $T_1\cup\cdots\cup T_r = \{1,\ldots,t\}$ be the partition defined
  by
  \begin{equation}
    T_j = \{k\in\{1,\ldots,t\}\,:\,D_k(j,j)\not=0\},
  \end{equation}
  we find that $V_k \sim U_j$ for every $k\in T_j$ and
  \begin{equation}
    p_j = \sum_{k\in T_j} q_k.
  \end{equation}
  The channel $\Phi$ therefore has a unique mixed-unitary decomposition.
\end{proof}

We may now use Theorem \ref{thm:r2r} to construct  mixed-unitary channels in
dimension $n=p+1$ having mixed-unitary rank strictly greater than their Choi
rank for any odd prime $p$.
The channels constructed in this manner will be shown to have Choi rank $p+1$
but mixed-unitary rank $2p$, yielding an increasingly large separation between
mixed-unitary rank and Choi rank as $p$ increases.
Moreover, the channels constructed in this manner are not unitary equivalent to
a Schur map (see Appendix \ref{appendix:schurmaps}).
This construction makes use of the discrete Weyl matrices. 

\begin{example}\label{ex:oddprimeWeylchannel}
  Let $p$ be an odd prime integer.
  Define $\zeta = \exp(2\pi i/p)$ and define unitary matrices
  $U,V\in\U_p$ as 
  \begin{equation}\label{eq:generalizedPauli}
    U= \sum_{a\in\mathbb{Z}_p} E_{a+1,a}
    \qquad\text{and}\qquad
    V = \sum_{a\in\mathbb{Z}_p}\zeta^{a} E_{a,a},
  \end{equation}
  where one takes $\{e_a\,:\, a\in\mathbb{Z}_p\}$ as the standard basis of
  $\complex^p$, and define a mixed-unitary channel $\Phi:\M_p\rightarrow\M_p$ as
  \begin{equation}\label{eq:ambinissmithchannel}
    \Phi(X) = \frac{1}{p}\sum_{a\in\mathbb{Z}_p}
    \bigl(U^aV^{a^2}\bigr)X\bigl(U^aV^{a^2}\bigr)^*
  \end{equation}
  for each $X\in\M_p$.
  The collection of unitary matrices $\{U^aV^b\,:\, a,b\in\mathbb{Z}_p\}$ form
  an orthogonal basis of $\M_p$, and these matrices satisfy
  \begin{equation}
    (U^aV^b)^*(U^cV^d)  \sim U^{c-a}V^{d-b}
  \end{equation}
  for each $a,b,c,d\in\mathbb{Z}_p$.
  It is evident that the collection $\{U^aV^{a^2}\,:\, a\in\mathbb{Z}_p\}$
  is linearly independent, and thus $\Phi$ has Choi rank and mixed-unitary rank
  both equal to $p$.
  We will show that the dimension of the operator system of $\Phi$ satisfies
  $\op{dim}(\mathcal{S}_\Phi) = p^2-p+1$.
  To prove this claim, we will show that, for any $a,b,c,d\in\mathbb{Z}_p$, the
  matrices
  \begin{equation}\label{eq:UVabcd}
    \bigl(U^bV^{b^2}\bigr)^*\bigl(U^aV^{a^2}\bigr)
    \qquad\text{and}\qquad
    \bigl(U^dV^{d^2}\bigr)^*\bigl(U^cV^{c^2}\bigr)
  \end{equation}
  are orthogonal unless at least one of $(a,b)=(c,d)$ or $(a,c)=(b,d)$ holds.
  Indeed, note that
  \begin{equation}
    \begin{aligned}
      \bigabs{\bigip{\bigl(U^bV^{b^2}\bigr)^*\bigl(U^aV^{a^2}\bigr)}{
      \,\bigl(U^dV^{d^2}\bigr)^*\bigl(U^cV^{c^2}\bigr)}}
      & = \bigabs{\bigip{U^{a-b}V^{a^2-b^2}}{U^{c-d}V^{c^2-d^2}}} \\
      &= \begin{cases}
        p & \text{if }a-b= c-d\text{ and }a^2-b^2=c^2-d^2\\
        0 & \text{otherwise},
      \end{cases}
    \end{aligned}
  \end{equation}
  where the equalities are taken to be equivalences modulo $p$.
  Suppose now that the pair of matrices in \eqref{eq:UVabcd} are not orthogonal
  and suppose further that $a\neq b$.
  As it must be the case that 
  \begin{equation}
    a-b=c-d \qquad\text{and}\qquad a^2-b^2=c^2-d^2,
  \end{equation}
  where $a-b\neq0$, we may divided the second equality by the first to find
  that
  \begin{equation}
    a-b=c-d \qquad\text{and}\qquad a+b=c+d.
  \end{equation}
  Taking both the sum and difference of these two resulting equalities, we find
  that
  \begin{equation}
    2a=2c\qquad\text{and}\qquad 2b=2d.
  \end{equation}
  As these equalities are taken to be equivalences modulo $p$ (where $p$ is an
  odd prime), we may conclude that $a=c$ and $b=d$.
  This completes the proof of the claim that
  $\op{dim}(\mathcal{S}_\Phi) = p^2-p+1$.
  It follows that $\Phi$ has a unique mixed-unitary decomposition by
  Theorem~\ref{thm:uniqueMU}.

  Now let $\Psi:\M_1\rightarrow\M_1$ be the (rather trivial) channel defined as
  $\Psi(\alpha)=\alpha$ for every $\alpha\in\complex$.
  By Theorem \ref{thm:r2r}, the channel
  $\Phi\oplus\Psi:\M_{p+1}\rightarrow\M_{p+1}$ is mixed unitary with Choi rank
  equal to $p+1$ but mixed-unitary rank equal to $2p$.
\end{example}

\begin{remark}
 We remark that the channel in \eqref{eq:ambinissmithchannel} appears in
 \cite{AmbinisS2004} in the context of approximate quantum encryption schemes.
\end{remark}

\begin{example}\label{ex:IXZmixedunitary}
  In order to provide a concrete example, we now explicitly present the mixed
  unitary channel from Example \ref{ex:oddprimeWeylchannel} in the case when
  $p=3$, where $\zeta = \exp(2\pi i /3)$.
  The matrices in \eqref{eq:generalizedPauli} are
  \begin{equation}
    U = \begin{pmatrix}
      0 & 0 & 1 \\
      1 & 0 & 0\\
      0 & 1 & 0
    \end{pmatrix}
    \quad\text{and}\quad
    V = \begin{pmatrix}
      1 & 0 & 0 \\
      0 & \zeta & 0\\
      0 & 0 & \zeta^2
    \end{pmatrix},
  \end{equation}
  and the channel $\Phi:\M_3\rightarrow\M_3$ as defined in
  \eqref{eq:ambinissmithchannel} is given by
  \begin{equation}
    \Phi(X) = \frac{1}{3} \bigl(W_0XW_0^* + W_1XW_1^*+ W_2XW_2^*\bigr),
  \end{equation}
  where one defines the unitary matrices $W_a=U^aV^{a^2}$ for each
  $a\in\{0,1,2\}$.
  Explicitly,
  \begin{equation}
    W_0 = \begin{pmatrix}
      1 & 0 & 0 \\
      0 & 1 & 0\\
      0 & 0 & 1
    \end{pmatrix},
    \quad
    W_1 = \begin{pmatrix}
      0 & 0 & \zeta ^2 \\
      1 & 0 & 0 \\
      0 & \zeta  & 0 
    \end{pmatrix},
    \quad\text{and}\quad
    W_2 = \begin{pmatrix}
      0 &\zeta & 0 \\ 0 & 0 &\zeta^2 \\ 1 & 0 & 0 
    \end{pmatrix}.
  \end{equation}
  The Choi rank of $\Phi$ is
  $\op{rank}(J(\Phi))=3$ and $\Phi$ has mixed-unitary rank equal to 3.
  The operator system $\mathcal{S}_{\Phi}$ is spanned by the seven linearly
  independent matrices:
  \begin{equation}
    \begin{gathered}
      W_0^*W_1 =
      \begin{pmatrix}
        0 & 0 &\zeta^2 \\ 
        1 & 0 & 0 \\ 
        0 &\zeta & 0
      \end{pmatrix},
      \quad
      W_0^*W_2 =
      \begin{pmatrix}
        0 & \zeta  & 0 \\
        0 & 0 & \zeta ^2 \\
        1 & 0 & 0
      \end{pmatrix}
      \quad
      W_1^*W_0 =
      \begin{pmatrix}
        0 & 1 & 0 \\
        0 & 0 & \zeta^2\\
        \zeta & 0 & 0
      \end{pmatrix},\\[2mm]
      W_1^*W_2 =
      \begin{pmatrix}
        0 & 0 & \zeta^2 \\
        \zeta^2 & 0 & 0\\
        0 & \zeta^2 & 0
      \end{pmatrix},
      \quad 
      W_2^*W_0 =
      \begin{pmatrix}
        0 & 0 & 1 \\
        \zeta^2 & 0 & 0\\
        0 & \zeta & 0
      \end{pmatrix},
      \quad
      W_2^*W_1 =
      \begin{pmatrix}
        0 & \zeta & 0 \\
        0 & 0 & \zeta\\
        \zeta & 0 & 0
      \end{pmatrix},\\[2mm]
      W_0^*W_0 = W_1^*W_1 = W_2^*W_2=
      \begin{pmatrix}
        1 & 0 & 0 \\
        0 & 1 & 0\\
        0 & 0 & 1
      \end{pmatrix},
    \end{gathered}
  \end{equation}
  and thus $\op{dim}(\mathcal{S}_{\Phi})=7$.
  It follows from Theorem~\ref{thm:uniqueMU} that $\Phi$ has a unique
  mixed-unitary decomposition.
  Defining the trivial channel $\Psi:\M_1\rightarrow\M_1$ as
  $\Psi(\alpha)=\alpha$ for every $\alpha\in\complex$, it follows from Theorem
  \ref{thm:r2r} that the channel $\Phi\oplus\Psi:\M_4\rightarrow\M_4$ is mixed
  unitary with Choi rank equal to 4 but mixed-unitary rank equal to 6.
  Explicitly, this channel is given by
  \begin{equation}
    (\Phi\oplus\Psi)(X) = \frac{1}{6}\sum_{k=1}^6 A_kXA_k^*,
  \end{equation}
  where $A_1,\dots,A_6\in\U_4$ are the unitary matrices defined as
  \begin{equation}
    \begin{gathered}
      A_1 = \begin{pmatrix}
        1 & 0 & 0 &0\\
        0 & 1 & 0 &0\\
        0 & 0 & 1 & 0\\
        0 & 0 & 0 & 1
      \end{pmatrix},
      \quad
      A_2 = \begin{pmatrix}
        1 & 0 & 0 &0\\
        0 & 1 & 0 &0\\
        0 & 0 & 1 & 0\\
        0 & 0 & 0 & -1
      \end{pmatrix},\quad
      A_3 = \begin{pmatrix}
        0 & 0 & \zeta^2 &0\\
        1 & 0 & 0 &0\\
        0 & \zeta & 0 & 0\\
        0 & 0 & 0 & 1
      \end{pmatrix},\\[2mm]  
      A_4 = \begin{pmatrix}
        0 & 0 & \zeta^2 &0\\
        1 & 0 & 0 &0\\
        0 & \zeta & 0 & 0\\
        0 & 0 & 0 & -1
      \end{pmatrix}, \quad
      A_5 = \begin{pmatrix}
        0 & \zeta & 0 &0\\
        0 & 0 & \zeta^2 &0\\
        1 & 0 & 0 & 0\\
        0 & 0 & 0 & 1
      \end{pmatrix},
      \quad
      A_6 = \begin{pmatrix}
        0 & \zeta & 0 &0\\
        0 & 0 & \zeta^2 &0\\
        1 & 0 & 0 & 0\\
        0 & 0 & 0 & -1
      \end{pmatrix}.
    \end{gathered}
  \end{equation}
\end{example}

\section{Further examples based on Schur channels}
\label{sec:corr}

Further examples illustrating properties of the mixed-unitary rank are
presented in this section.
These examples fall into the category of \emph{Schur channels}, which are
channels that can be expressed as
\begin{equation}
  \label{eq:Schur-map}
  \Phi(X) = C\odot X
\end{equation}
for every $X\in\M_n$, for some fixed choice of $C\in\M_n$, where
$C\odot X$ denotes the \emph{Schur product} (or entry-wise product) of the
matrices $C$ and $X$.
Schur channels are sometimes alternatively called \emph{diagonal channels},
owing to the fact that every Kraus representation of a Schur channel must make
use of only diagonal Kraus matrices.
The Choi rank of the Schur channel \eqref{eq:Schur-map} is given by
$\op{rank}(J(\Phi))=\op{rank}(C)$, and it is well known that a map of this form is a
channel if and only if $C$ is a \emph{correlation matrix}, which is a
positive semidefinite matrix whose diagonal entries are all equal to 1.
Every Schur channel is necessarily unital; meanwhile, for $n\geq 4$, 
there are examples of Schur
channels that are not mixed unitary
\cite{Tregub1986,Landau1993}.

Before proceeding to the examples promised, it will be helpful to note
various properties of Schur channels, and mixed-unitary Schur channels in
particular.
First, we observe that the dimension of the operator system of any Schur
channel can be calculated directly from the formula
\begin{equation}
  \op{dim}(\mathcal{S}_\Phi) = \op{rank}(\overline{C}\odot C),
\end{equation}
which follows from the fact that
$(\Phi^*\Phi)(X) = (\overline{C}\odot C)\odot X$ for every $X\in\M_n$.
We also note that the operator system of every Schur channel contains only
diagonal matrices.

Second, we observe that the mixed-unitary rank of a mixed-unitary Schur
channel can alternatively be characterized directly in terms of what we call
the \emph{toroidal rank} of the matrix $C$.
To be precise, let us introduce the notation
\begin{equation}
 \mathbb{T} = \{\alpha\in\complex\, :\, \abs{\alpha} = 1\}.
\end{equation}
It is evident that a correlation matrix $C\in\M_n$ has rank equal to 1 if and
only if $C = u u^{\ast}$ for some choice of a vector $u\in\mathbb{T}^n$.
We shall say that a correlation matrix $C$ is \emph{toroidal} if it can be
expressed as a convex combination of rank-one correlation matrices.
That is, $C\in\M_n$ is toroidal if there exists a positive integer $N$,
vectors $u_1,\ldots,u_N\in\mathbb{T}^n$, and a probability vector
$(p_1,\ldots,p_n)$ such that
\begin{equation}\label{eq:Cdecomp}
  C = \sum_{k=1}^N p_k u_k u_k^*.
\end{equation}
The \emph{toroidal rank} of $C$ is the smallest positive integer $N$ for which
such an expression exists.
The observation that the mixed-unitary rank of the Schur channel
\eqref{eq:Schur-map} coincides with the toroidal rank of $C$ is expressed by
the following proposition.

\begin{proposition}
  Let $n$ and $N$ be positive integers, let $C\in\M_n$ be a correlation matrix,
  and let $\Phi$ be the Schur channel defined as $\Phi(X) = C \odot X$ for
  every $X\in\M_n$.
  The following two statements are equivalent:
  \begin{enumerate}
  \item[1.]
    $\Phi$ is mixed unitary and has mixed-unitary rank equal to $N$.
  \item[2.]
    $C$ is toroidal and has toroidal rank equal to $N$.
  \end{enumerate}
\end{proposition}

For all dimensions $n\geq 4$ there exist correlation matrices
in $\M_n$ that are not toroidal \cite{Li1994}.
However, it is the case that every correlation matrix in $\M_2$ and $\M_3$ is
toroidal.
Indeed, it follows from Theorem \ref{thm:lowdimension} that all correlation
matrices in $\M_2$ and $\M_3$ have toroidal ranks equal to their ranks.

\begin{proposition}\label{thm:Cminrank23}
  Let $n\in\{2,3\}$.
  Every correlation matrix $C\in\M_n$ is toroidal and has toroidal rank equal
  to $\op{rank}(C)$.
  Equivalently, the channel $\Phi:\M_n\rightarrow\M_n$ defined as
  $\Phi(X)=C\odot X$ for each $X\in\M_n$ is mixed unitary and has
  mixed-unitary rank equal to its Choi rank.
\end{proposition}

\begin{proof}
  The operator system $\mathcal{S}_\Phi\subseteq\M_n$ consists of only diagonal
  matrices and thus $\op{dim}(\mathcal{S})\leq n$.
  The result now follows from Theorem \ref{thm:lowdimension}, as
  $\op{rank}(J(\Phi))=\op{rank}(C)$ and we have assumed that $n\leq 3$. 
\end{proof}

\noindent
Theorem \ref{thm:upperbound} implies the following upper bound on the toroidal
rank of any correlation matrix.

\begin{corollary}
  Let $n$ be a positive integer and let $C\in\M_n$ be a toroidal correlation
  matrix having toroidal rank $N$.
  It is the case that
  \begin{equation}
    N \leq r^2 - s + 1,
  \end{equation}
  where $r=\op{rank}(C)$ and $s=\op{rank}(\overline{C}\odot C)$.
\end{corollary}

Now we are prepared to proceed to the examples suggested previously.
The following lemma will be used for the first example.

\begin{lemma} \label{lem:3corrunique}
  Let $C\in\M_3$ be a correlation matrix with $\op{rank}(C)=2$, and assume
  that none of the off-diagonal entries of $C$ is contained in $\mathbb{T}$.
  It must then be the case that $\op{rank}(\overline{C}\odot C) = 3$.
\end{lemma}

\begin{proof}
  Note first that $C$ is a toroidal correlation matrix, as every $3\times 3$
  correlation matrix is toroidal.
  By Proposition~\ref{thm:Cminrank23}, the toroidal rank of $C$ must be
  equal to 2, so there must exist vectors $u_0,u_1\in\mathbb{T}^3$ such that
  $C\in\op{conv}\{u_0u_0^*,u_1u_1^*\}$.
  It may be assumed without loss of generality that 
  \begin{equation}
    u_0 = \begin{pmatrix}1\\\alpha_0\\\beta_0\end{pmatrix}
      \quad\text{and}\quad
      u_1 = \begin{pmatrix}1\\\alpha_1\\\beta_1\end{pmatrix}
  \end{equation}
  for some choice of complex units
  $\alpha_0,\alpha_1,\beta_0,\beta_1\in\mathbb{T}$.
  By the assumption that none of the off-diagonal entries of $C$ lies in
  $\mathbb{T}$,  we have
  \begin{equation}
    \alpha_0\neq\alpha_1, \quad
    \beta_0\neq\beta_1, \quad \text{and}\quad
    \overline{\alpha_0}\beta_0 \neq \overline{\alpha_1}\beta_1.
  \end{equation}
  Define the complex units $\alpha,\beta\in\mathbb{T}$ as
  $\alpha=\alpha_0\overline{\alpha_1}$ and $\beta=\beta_0\overline{\beta_1}$,
  and observe that $1\not\in\{\alpha,\beta,\overline{\alpha}\beta\}$.
  It is the case that
  \begin{equation}
    \begin{aligned}
      \op{im}(\overline{C}\odot C) &=
      \op{span}\{u_0\odot\overline{u_0},
      u_0\odot\overline{u_1},u_1\odot\overline{u_0},
      u_1\odot\overline{u_1}\}\\
      & = \op{span}\left\{
      \begin{pmatrix}1 \\ 1 \\ 1\end{pmatrix},
        \begin{pmatrix}1 \\ \alpha \\ \beta \end{pmatrix},
        \begin{pmatrix} 1 \\ \overline{\alpha} \\ \overline{\beta}
        \end{pmatrix}\right\}.
    \end{aligned}
  \end{equation}

  Now define a matrix $A\in\M_3$ as
  \begin{equation}
   A = \begin{pmatrix}
      1 & 1 & 1\\
      1 & \alpha & \overline{\alpha} \\
      1 & \beta & \overline{\beta}
    \end{pmatrix}, 
  \end{equation}
  for which it may be verified (using the fact that
  $\alpha,\beta\in\mathbb{T}$) that 
  \begin{equation}
    \begin{aligned}
      \det(A^*A) &=
      -(2-\alpha-\overline{\alpha})
      (2-\beta-\overline{\beta})
      (2-\alpha\overline{\beta}-\overline{\alpha}\beta)\\
      & = 8\op{Re}(\alpha-1)\op{Re}(\beta-1)\op{Re}(\overline{\alpha}\beta-1).
    \end{aligned}
  \end{equation}
  As $\alpha,\beta,\overline{\alpha}\beta\in\mathbb{T}$ but none of $\alpha$,
  $\beta$, and $\overline{\alpha}\beta$ is equal to 1, it follows that
  $\det(A^*A)\neq0$ and therefore $A$ is nonsingular.
  This implies that the columns of $A$ are linearly independent, and therefore
  we have $\op{dim}(\op{im}(\overline{C}\odot C))=3$, as required.
\end{proof}

\begin{example}\label{ex:C4x4}
  Define the (necessarily toroidal) correlation matrix $B\in\M_3$ as
  \begin{equation}
    B = \begin{pmatrix}
      1 & \frac{1}{\sqrt{2}} & \frac{1}{\sqrt{2}}  \\
      \frac{1}{\sqrt{2}} & 1 & 0 \\
      \frac{1}{\sqrt{2}} & 0 & 1 
    \end{pmatrix}.
  \end{equation}
  Note that $\op{rank}(B)=2$ and that none of its off-diagonal entries lies in
  $\mathbb{T}$.
  It follows from Lemma~\ref{lem:3corrunique} that
  $\op{rank}(\overline{B}\odot B) = 3$.
  The channel defined as $\Phi(X) = B \odot X$ for all $X\in\M_3$
  therefore has Choi rank $r = 2$ and an operator system of dimension $s=3$.
  By Theorem~\ref{thm:uniqueMU}, it follows that $\Phi$ has mixed-unitary rank
  equal to $r=2$ and has a unique mixed-unitary decomposition.
  Equivalently, $B$ has toroidal rank $2$ and has a unique toroidal
  decomposition.
  The fact that $B$ has toroidal rank 2 may also be observed directly from the
  toroidal decomposition
  \begin{equation}
    B=\frac{1}{2} uu^* + \frac{1}{2} vv^*
  \end{equation}
  for the choice of toroidal vectors $u,v\in\mathbb{T}^3$ given by
  \begin{equation}
    u=\begin{pmatrix}
    1\\[2mm]
    \frac{1+i}{\sqrt{2}}\\[3mm]
    \frac{1-i}{\sqrt{2}}
    \end{pmatrix}
    \quad\text{and}\quad
    v=\begin{pmatrix}
    1\\[2mm]
    \frac{1-i}{\sqrt{2}}\\[3mm]
    \frac{1+i}{\sqrt{2}}
    \end{pmatrix}.
  \end{equation}
  
  Now consider the channel $\Phi \oplus \Psi$, where
  $\Psi:\M_n \rightarrow \M_n$ is the identity channel for any choice of a
  dimension $n \geq 1$.
  By Theorem~\ref{thm:r2r}, this channel has Choi rank 3 and mixed-unitary
  rank 4.
  This direct sum channel is also a Schur channel, owing to the fact that
  the identity channel is a Schur channel corresponding to the all 1s matrix.
  In particular, for $n=1$ we find that the correlation matrix $C\in\M_4$
  defined as
  \begin{equation}
    \label{eq:C4x4}
    C = \begin{pmatrix}
      1 & \frac{1}{\sqrt{2}} & \frac{1}{\sqrt{2}} & 0 \\
      \frac{1}{\sqrt{2}} & 1 & 0 & 0\\
      \frac{1}{\sqrt{2}} & 0 & 1 & 0\\
      0 & 0 & 0 & 1
    \end{pmatrix},
  \end{equation}
  has $\op{rank}(C)=3$ and toroidal rank equal to 4.
\end{example}

For the next example we will require the notion of
\emph{mutually unbiased bases}, which is as follows.
Let $d$ be a positive integer and let
$\mathcal{A}_1,\dots,\mathcal{A}_N\subset\complex^d$ be orthonormal bases of
$\complex^d$ given as
\begin{equation}
  \mathcal{A}_k=\{u_{k,1},\dots,u_{k,d}\}
\end{equation}
for each $k\in\{1,\dots,N\}$.
This collection of bases $\{\mathcal{A}_1,\dots,\mathcal{A}_N\}$ is said to be
\emph{mutually unbiased} if, for all choices of distinct indices
$i\not=j\in\{1,\dots,N\}$, it is the case that
\begin{equation}
  \abs{\ip{u}{v}} = \frac{1}{\sqrt{d}}
\end{equation}
for all $u\in\mathcal{A}_i$ and $v\in\mathcal{A}_j$.
An upper bound to the maximal size $N$ of a collection of mutually unbiased
bases that may exist in $\complex^d$ is $N \leq d+1$.
It is known that this bound is achieved in the case when $d$ is a prime power
(see, e.g., \cite{Ivonovic1981}), while it is a major open question to
determine if this maximum value can be achieved for non-prime-powers.
More information on mutually unbiased bases can be found in \cite{Wootters1989} and \cite{Durt2010}.

In the following example we will show how to construct correlation matrices
with rank $d+1$ and toroidal rank equal to $2d$ for any $d$ for which
$d+1$ mutually unbiased bases of $\complex^d$ exist.

\begin{example}
  Let $d$ be a positive integer and suppose that there exist $d+1$ mutually
  unbiased bases $\mathcal{A}_1,\dots,\mathcal{A}_{d+1}\subset\complex^d$ given as
  $\mathcal{A}_t=\{u_{t,1},\dots,u_{t,d}\}$ for each $t\in\{1,\dots,d+1\}$.
  Define a matrix $A\in\M_{d^2,d}$ as
  \begin{equation}
    A = \sum_{k, j=1}^d (e_k\otimes e_j)u_{k,j}^* 
  \end{equation}
  and define $C\in\M_{d^2}$ as $C = AA^*$.
  It is evident that $\op{rank}(C)=d$ and that $C$ is a correlation matrix.

  Let us first verify that $C$ is toroidal, with toroidal rank equal to $d$.
  Define $v_1,\dots,v_d\in\complex^d\otimes\complex^d$ as
  \begin{equation}
    v_k =\sqrt{d} \sum_{i,j}^d \ip{u_{i,j}}{u_{d+1,k}}
    e_i\otimes e_j
  \end{equation}
  for each $k\in\{1,\dots,d\}$, and observe that
  \begin{equation}
    \abs{v_k(i,j)} = \sqrt{d}\,\abs{\ip{u_{i,j}}{u_{d+1,k}}} = 1
  \end{equation}
  for all $i,j,k\in\{1,\dots,d\}$.
  By defining a unitary matrix $U\in\U_d$ as 
  \begin{equation}
    U = \sum_{k=1}^d u_{d+1,k}\hspace{0.2mm}e_k^*,
  \end{equation}
  one may verify that
  \begin{equation}
    \frac{1}{d}\sum_{k=1}^d v_kv_k^* = AUU^*A = AA^* = C.
  \end{equation}
  
  Now let us compute $\overline{C}\odot C$.
  For each choice of indices $i,j,k,\ell\in\{1,\ldots,d\}$, we may express
  the corresponding entry of $\overline{C}\odot C$ as follows:
  \begin{equation}
    \bigl( \overline{C}\odot C \bigr)\bigl( (k,i), (\ell,j))
    =
    \begin{cases}
      1 & \text{if $k = \ell$ and $i = j$}\\
      0 & \text{if $k = \ell$ and $i \not= j$}\\
      \frac{1}{d} & \text{if $k \not= \ell$}.
    \end{cases}
  \end{equation}
  As a block matrix, $\overline{C}\odot C$ takes this form:
  \begin{equation}
    \overline{C}\odot C
    =
    \begin{pmatrix}
      \I_d & \frac{1}{d} J_d & \cdots & \frac{1}{d}J_d\\[1mm]
      \frac{1}{d} J_d & \I_d & \ddots & \vdots \\[1mm]
      \vdots & \ddots & \ddots & \frac{1}{d}J_d\\[1mm]
      \frac{1}{d} J_d & \cdots & \frac{1}{d}J_d & \I_d
    \end{pmatrix},
  \end{equation}
  where $J_d$ denotes the $d\times d$ matrix having a 1 in every entry.
  Equivalently,
  \begin{equation}
    \overline{C}\odot C
    = \frac{1}{d} J_d \otimes J_d + \I_d \otimes \bigl(\I_d -
    \frac{1}{d}J_d\bigr).
  \end{equation}
  As $\I_d - J_d/d$ and $J_d/d$ are orthogonal projection matrices of
  rank $d-1$ and $1$, respectively, we conclude that
  \begin{equation}
    \op{rank}\bigl(\overline{C}\odot C\bigr) = 1 + d(d-1) = d^2 - d + 1.
  \end{equation}

  Through a similar argument to the previous example, we conclude that
  if $\Phi:\M_{d^2} \rightarrow \M_{d^2}$ is the Schur channel defined as
  \begin{equation}
    \Phi(X) = C \odot X
  \end{equation}
  for all $X\in\M_{d^2}$ and $\Psi$ is a unitary channel of any dimension,
  then the channel $\Phi\oplus \Psi$ is a mixed-unitary channel having
  Choi rank $d+1$ and mixed-unitary rank $2d$.
\end{example}

Our final example reveals that the mixed-unitary rank is not multiplicative
with respect to tensor products.

\begin{example}
  Let $C\in\M_4$ be the correlation matrix as defined in \eqref{eq:C4x4}.
  This correlation matrix has $\op{rank}(C)=3$ and toroidal rank equal to $4$.
  However, the correlation matrix $C\otimes \I_2$ satisfies
  $\op{rank}(C\otimes\I_2)=6$ and has toroidal rank also equal to 6.

  To see this, one may construct a toroidal decomposition of $C\otimes\I_2$ as
  follows.
  Define a $6\times 8$ matrix $A$ as
  \begin{equation}
    A = \left(\begin{array}{cccccccc}
      0&3&-3&0&0&3&-3&0\\
      0&-3&3&12&12&9&-9&0\\
      8&11&5&-8&0&3&-3&-8\\
      0&3&-3&-8&8&11&5&-8\\
      0&-3&3&-4&4&1&7&8\\
      4&1&7&8&0&-3&3&-4
    \end{array}\right)
  \end{equation}
  and define vectors $u_1,\dots,u_6\in\mathbb{T}^8$ as
  \begin{equation}
    u_j(k) = \exp(2\pi i A(j,k)/24)
  \end{equation}
  for each $j\in\{1,\dots,6\}$ and $k\in\{1,\dots,8\}$.
  It may be verified (most easily with the help of a computer) that
  \begin{equation}
    C\otimes \I_2 = \frac{1}{6}\sum_{j=1}^6 u_ju_j^*.
  \end{equation} 

  Thus, by taking $\Phi\in\M_4\rightarrow\M_4$ to be the Schur channel defined
  by $\Phi(X)=C\odot X$ for each $X\in\M_4$, and letting
  $\Delta:\M_2\rightarrow\M_2$ be the \emph{completely dephasing channel}, which
  is the Schur channel given by
  \begin{equation}
    \Delta(Y) = \I_2 \odot Y
  \end{equation}
  for all $Y\in\M_2$, one finds that the mixed-unitary rank of
  $\Phi\otimes\Delta$ is 6, despite the fact that the mixed-unitary ranks of
  $\Phi$ and $\Delta$ are 4 and 2, respectively.
\end{example}

\section{Mixed-unitary rank of Werner--Holevo channels}
\label{sec:wernerholevo}

The \emph{Werner--Holevo channels} are interesting examples of unital channels defined as
\begin{equation}
 \Phi_0(X) = \frac{\Tr(X)\I_n + X^\t}{n+1}
 \qquad\text{ and }\qquad
  \Phi_1(X) = \frac{\Tr(X)\I_n - X^\t}{n-1}
\end{equation}
for each $X\in\M_n$. We will call $\Phi_0$ the \emph{symmetric} Werner--Holevo channel and $\Phi_1$ the \emph{anti-symmetric} Werner--Holevo channel. For these channels, one has
\begin{equation}
 J(\Phi_0) = \frac{2}{n+1}\Pi_0 
 \qquad\text{ and }\qquad
  J(\Phi_1) = \frac{2}{n-1}\Pi_1 
\end{equation}
where $\Pi_0$ and $\Pi_1$ are the projection matrices onto the symmetric and
anti-symmetric subspaces of $\complex^n\otimes\complex^n$ respectively.
The Werner--Holevo channels have Choi ranks equal to
\begin{equation}
 \op{rank}(J(\Phi_0)) = \binom{n+1}{2} = \frac{n(n+1)}{2} \qquad \text{and}\qquad\op{rank}(J(\Phi_1)) = \binom{n}{2} = \frac{n(n-1)}{2}
\end{equation}
respectively. It is known that $\Phi_1$ is not mixed unitary for any odd $n$. It is perhaps known that $\Phi_0$ is mixed unitary for all $n$ and that $\Phi_1$ is mixed unitary for all even $n$. In this section we will present mixed-unitary decompositions showing that both $\Phi_0$ and $\Phi_1$ have minimal mixed-unitary rank for all even $n$. For $n=3$, the symmetric Werner--Holevo channel $\Phi_0$ also has minimal mixed-unitary rank and we conjecture based on numerical evidence that $\Phi_0$ has minimal mixed-unitary rank for all odd $n$ as well.

Before proceeding with the presentation of the mixed-unitary decompositions of
the Werner--Holevo channels, allow us to first remark on the relationship
between the Werner--Holevo channels and the spaces of symmetric and
skew-symmetric matrices.
Denote the spaces of symmetric matrices $S_n\subset\M_n$ and skew-symmetric
matrices $K_n\subset\M_n$ as
\begin{equation}
 S_n = \{A\in\M_n\, :\, A^\t = A\}
 \qquad\text{and}\qquad
 K_n = \{A\in\M_n\, :\, A^\t = -A\}.
\end{equation}
These spaces have dimensions $\op{dim}(S_n) = \binom{n+1}{2}$ and $\op{dim}(K_n) = \binom{n}{2}$ respectively.
The subspaces of $\complex^n\otimes\complex^n$ onto which the symmetric projection matrix
$\Pi_0$ and anti-symmetric projection matrix $\Pi_1$ project are precisely 
\begin{equation}
 \op{im}(\Pi_0) = \{\vec(A)\, :\, A\in S_n\} 
 \qquad\text{and}\qquad
 \op{im}(\Pi_1) = \{\vec(A)\, :\, A\in K_n\}.
\end{equation}
Moreover, if $\Pi\in\M_m$ is any projection matrix with $\op{rank}(\Pi)=r $ and
$x_1,\dots,x_r\in\complex^m$ are any vectors, it holds that $\Pi= \sum_{k=1}^r
x_kx_k^*$ if and only if $\{x_1,\dots,x_r\}$ is an orthonormal basis for
$\op{im}(\Pi)$.
This allows us to make the following observation.

\begin{theorem}\label{thm:wernerholevobases}
 Let $n$ be a positive integer. The following statements hold. 
 \begin{enumerate}
  \item The symmetric Werner--Holevo channel $\Phi_0$ has mixed-unitary rank equal to $\binom{n+1}{2}$ if and only if there exists an orthogonal basis of $S_n$ consisting of only unitary matrices.
  \item Suppose $n$ is even. The anti-symmetric Werner--Holevo channel $\Phi_1$ has mixed-unitary rank equal to $\binom{n}{2}$ if and only if there exists an orthogonal basis of $K_n$ consisting of only unitary matrices.
 \end{enumerate}
\end{theorem}

\begin{proof}
  We prove statement (1).
  The proof of statement (2) is analogous.
  Suppose there exist unitary matrices $U_1,\dots,U_{n(n+1)/2}\subset \U_n$ and
  a probability vector $(p_1,\dots,p_{n(n+1)/2})$ satisfying
 \begin{equation}
   \Phi_0(X) = \sum_{k=1}^{\frac{n(n+1)}{2}}p_k U_kXU_k^*
 \end{equation}
 for each $X\in\M_n$. It holds that 
 \begin{equation}
   \Pi_0 = \frac{n+1}{2}J(\Phi_0) = \sum_{k=1}^{\frac{n(n+1)}{2}}\frac{n+1}{2}p_k \vec(U_k)\vec(U_k)^*.
 \end{equation}
 It follows that $p_k=2/n(n+1)$ for each $k\in\{1,\dots,n(n+1)/2\}$ and that
 the collection of unitary matrices $\{U_1,\dots,U_{n(n+1)/2}\}\subset\U_n$ is
 an orthogonal basis for $S_n$.
 The reverse implication is immediate.
\end{proof}

The remainder of this section is dedicated to constructing mixed-unitary
decompositions of the Werner--Holevo channels.
We will introduce the following notation.
For each positive integer $n$, the space of matrices $\M_n$ is spanned by the
collection of Hermitian matrices
\begin{equation}
 \{H_{j,k}\,:\, j,k\in\{1,\dots,n\}\}
\end{equation}
defined by
\begin{equation}
 H_{j,k} = \left\{\begin{array}{ll}
                   E_{j,j} & \text{if }j=k\\
                   \frac{1}{\sqrt{2}}(E_{j,k}+E_{k,j}) & \text{if }j<k\\
                   \frac{1}{\sqrt{2}}(iE_{j,k}-iE_{k,j}) & \text{if }j>k
                  \end{array} \right.
\end{equation}
for each pair of indices $j,k\in\{1,\dots,n\}$. It may be easily verified that the action of the Werner--Holevo channels can be given by
\begin{equation}
  \Phi_0(X) = \frac{2}{n+1}\sum_{1\leq j \leq k \leq n} H_{j,k} X H_{j,k}
  \qquad\text{and}\qquad
  \Phi_1(X) = \frac{2}{n-1}\sum_{1\leq k <j \leq n} H_{j,k} X H_{j,k}
\end{equation}
for each $X\in\M_n$. The construction of the mixed-unitary decompositions of the Werner--Holevo channels presented in the following will make use of the fact that the complete graph on $n$ vertices can be partitioned into $n-1$ disjoint perfect matchings for all even integers $n$ (see \cite{Hartman1985}). 

\subsection{Mixed-unitary rank of anti-symmetric Werner--Holevo channel}

For odd integers $n$, the anti-symmetric Werner--Holevo channel $\Phi_1$ is not mixed unitary. Here we show that, for even integers $n$, the anti-symmetric Werner--Holevo channel has mixed-unitary 
rank equal to its Choi rank. 

\begin{theorem}\label{thm:antisymWH}
 For each even positive integer $n$, the anti-symmetric Werner--Holevo channel $\Phi_1$ is mixed unitary with mixed-unitary rank equal to $\op{rank}(J(\Phi_1))=n(n-1)/2$.
\end{theorem}
\begin{proof}
 Consider the complete graph of $n$ vertices with vertices labelled
 $\{1,\dots,n\}$.
 The edge set of this graph may be identified with the set
 \begin{equation}
   \mathcal{E} = \{H_{j,k}\,:\, j,k\in\{1,\dots,n\} \text{ with }k<j\},
 \end{equation}
 where, for each $j,k\in\{1,\dots,n\}$ with $k<j$, the matrix $H_{j,k}$
 represents the edge connecting vertices $j$ and $k$.
 The edge set of this graph may be partitioned into $n-1$ disjoint perfect
 matchings $\mathcal{E}_1,\dots,\mathcal{E}_{n-1}$ such that
 \begin{equation}
   \mathcal{E} = \mathcal{E}_1\cup \mathcal{E}_2\cup\dots\cup \mathcal{E}_{n-1}.
 \end{equation}
 For each $\ell\in\{1,\dots,n-1\}$, we may label the $n/2$ elements of the
 perfect matching $\mathcal{E}_\ell$ as
 \begin{equation}
  \mathcal{E}_\ell = \{F_{\ell,1},\dots,F_{\ell,n/2}\}
 \end{equation}
 such that the matrix $F_{\ell,1} + \cdots + F_{\ell,n/2}$ has exactly one
 nonzero entry in each row and column.
 Setting $\zeta = \exp(2\pi i/n)$, we may define the matrices 
 \begin{equation}
  U_{\ell,a} = \sqrt{2}\sum_{b=1}^{\frac{n}{2}}\zeta^{2ab}F_{\ell,b} 
 \end{equation}
 for each pair of indices $\ell\in\{1,\dots,n-1\}$ and $a\in\{1,\dots,n/2\}$.
 It may be verified that each of the matrices $U_{\ell,a}$ is unitary, as each
 such matrix has exactly one nonzero entry in each row and column where each
 nonzero entry has modulus 1.
 Now, for each $X\in\M_n$ one has
 \begin{equation}
   \begin{aligned}
     \frac{2}{n(n-1)}\sum_{\ell=1}^{n-1}\sum_{a=1}^{\frac{n}{2}}
     U_{\ell,a}XU_{\ell,a}^*  
     &=  \frac{2}{n(n-1)}
     \sum_{\ell=1}^{n-1} \sum_{a=1}^{\frac{n}{2}} \; \sum_{b,c=1}^{\frac{n}{2}} 
     2 \, \zeta^{2a(b-c)}F_{\ell,b}XF_{\ell,c}^*   \\
     &=  \frac{2}{n-1} \sum_{\ell=1}^{n-1} \sum_{b=1}^{\frac{n}{2}} 
     F_{\ell,b}XF_{\ell,b}^*  \\
     &=\frac{2}{n-1}\sum_{1\leq k <j \leq n} H_{j,k} X H_{j,k} = \Phi_1(X),
   \end{aligned}
 \end{equation}
 and thus $\Phi_1$ can be expressed as the average of $\op{rank}(\Phi_1)=n(n-1)/2$
 unitary channels.
 It follows that $\Phi_1$ is mixed unitary with mixed-unitary rank equal to
 $n(n-1)/2$.
\end{proof}

\subsection{Mixed-unitary rank of symmetric Werner--Holevo channel}

\subsubsection{Symmetric Werner--Holevo channel for even \texorpdfstring{$n$}{n}}

For even integers $n$, the proof that the symmetric Werner--Holevo channel
$\Phi_0$ has minimal mixed-unitary rank is analogous to the proof for the
anti-symmetric version.

\begin{theorem}\label{thm:symmWHeven}
  For each positive even integer $n$, the symmetric Werner--Holevo channel
  $\Phi_0$ is mixed unitary with mixed-unitary rank equal to
  $\op{rank}(J(\Phi_0))=n(n+1)/2$.
\end{theorem}
\begin{proof}
  The proof is similar to the proof of Theorem \ref{thm:antisymWH}.
  As before, consider the complete graph of $n$ vertices with vertices labelled
  $\{1,\dots,n\}$, but here we identify the edge set of the graph with the
  collection of matrices
 \begin{equation}
   \mathcal{E} = \{H_{j,k}\,:\, j,k\in\{1,\dots,n\} \text{ with }j<k\},
 \end{equation}
 where, for each $j,k\in\{1,\dots,n\}$ with $j<k$, the matrix $H_{j,k}$
 represents the edge connecting vertices $k$ and $j$.
 The edge set of this graph may be partitioned into $n-1$ disjoint perfect
 matchings $\mathcal{E}_1,\dots,\mathcal{E}_{n-1}$ such that $\mathcal{E} =
 \mathcal{E}_1\cup \mathcal{E}_2\cup\dots\cup \mathcal{E}_{n-1}$.
 For each $\ell\in\{1,\dots,n-1\}$, we may label the $n/2$ elements of the
 perfect matching $\mathcal{E}_\ell$ as $\mathcal{E}_\ell =
 \{F_{\ell,1},\dots,F_{\ell,n/2}\}$.
 Setting $\zeta = \exp(2\pi i/n)$, we may define the matrices 
 \begin{equation}
   U_{\ell,a} = \sqrt{2}\sum_{b=1}^{\frac{n}{2}}\zeta^{2ab}F_{\ell,b} 
 \end{equation}
 for each pair of indices $\ell\in\{1,\dots,n-1\}$ and $a\in\{1,\dots,n/2\}$,  and define the matrices
 \begin{equation}
   V_j = \sum_{k=1}^n \zeta^{jk} H_{k,k}
 \end{equation}
 for each $j\in\{1,\dots,n\}$.
 It may be verified that each of the matrices $U_{\ell,a}$ and $V_j$ is
 unitary.
 Analogous to the proof of Theorem \ref{thm:antisymWH}, for each $X\in\M_n$ it
 holds that
 \begin{align}
   \frac{2}{n(n+1)}\sum_{\ell=1}^{n-1}\sum_{a=1}^{\frac{n}{2}}
   U_{\ell,a}XU_{\ell,a}^*
   & =
   \frac{2}{n(n+1)} \sum_{\ell=1}^{n-1}
   \sum_{a=1}^{\frac{n}{2}} \; \sum_{b,c=1}^{\frac{n}{2}}
   2 \; \zeta^{2a(b-c)}F_{\ell,b}XF_{\ell,c}^* 
    \nonumber\\
   &= 
    \frac{2}{n+1} \sum_{\ell=1}^{n-1} \sum_{b=1}^{\frac{n}{2}} 
   F_{\ell,b}XF_{\ell,b}^*
    \nonumber\\
   &=\frac{2}{n+1}\sum_{1\leq k <j \leq n} H_{j,k} X H_{j,k}.
   \label{eq:symmetricpart1}
 \end{align}
 One also has that
 \begin{align}
   \frac{2}{n(n+1)}\sum_{j=1}^nV_jXV_j^* 
   &= \frac{2}{n(n+1)}\sum_{j,k,\ell=1}^n
   \zeta^{j(k-\ell)} H_{k,k}XH_{\ell,\ell}\nonumber\\
   &=\frac{2}{n+1}\sum_{k=1}^n H_{k,k}X H_{k,k}\label{eq:symmetricpart2}
 \end{align}
 for each $X\in\M_n$.
 Putting together the results of \eqref{eq:symmetricpart1} and
 \eqref{eq:symmetricpart2}, we see that
 \begin{equation}
   \begin{aligned}
     \frac{2}{n(n+1)}\biggl(\sum_{\ell=1}^{n-1}\sum_{a=1}^{\frac{n}{2}}
     U_{\ell,a}XU_{\ell,a}^* + \sum_{j=1}^nV_jXV_j^*\biggr) 
     & =\frac{2}{n+1} \sum_{1\leq j\leq k\leq n} H_{j,k}X H_{j,k}\\
     & = \Phi_{0}(X)
   \end{aligned}
 \end{equation}
 holds for each $X\in\M_n$.
 As $\Phi_0$ is written as the average of $n(n+1)/2$ unitary channels, it
 follows that $\Phi_1$ is mixed unitary with mixed-unitary rank equal to
 $n(n+1)/2$.
\end{proof}

\subsubsection{Symmetric Werner--Holevo channel for odd \texorpdfstring{$n$}{n}}

For odd $n$, we will show that $\Phi_0$ has mixed-unitary rank at most $n(n+3)/2$. 

\begin{theorem}\label{thm:symmWHodd}
 For each odd positive integer $n$, the symmetric Werner--Holevo channel $\Phi_0$ has mixed-unitary rank at most $n(n+3)/2$.
\end{theorem}

\begin{proof}
 The proof is similar to the proofs of Theorems \ref{thm:antisymWH} and \ref{thm:symmWHeven}. Now however, consider the complete graph of $n+1$ vertices with vertices labelled $\{0,1,\dots,n\}$, and  identify the edge set of this graph with the collection of matrices defined as
 \begin{equation}
   \mathcal{E} = \{H_{j,k}\,:\, j,k\in\{0,1,\dots,n\} \text{ with }j<k\},
 \end{equation}
 where we define $H_{0,k}=H_{k,k}/\sqrt{2}$ for each $k\in\{1,\dots,n\}$. As before, this edge set may be partitioned into $n$ disjoint perfect matchings $\mathcal{E}_1,\dots,\mathcal{E}_{n}$ such that $\mathcal{E} = \mathcal{E}_1\cup \mathcal{E}_2\cup\dots\cup \mathcal{E}_{n}$. For each $\ell\in\{1,\dots,n\}$, we may label the $(n+1)/2$ elements of the perfect matching $\mathcal{E}_\ell$ as $\mathcal{E}_\ell = \{F_{\ell,1},\dots,F_{\ell,(n+1)/2}\}$. Setting $\zeta = \exp(2\pi i/(n+1))$, we may define the matrices 
 \begin{equation}
   U_{\ell,a} = \sqrt{2}\sum_{b=1}^{\frac{n+1}{2}}\zeta^{2ab}F_{\ell,b} 
 \end{equation}
 for each pair of indices $\ell\in\{1,\dots,n\}$ and $a\in\{1,\dots,(n+1)/2\}$, and the matrices
 \begin{equation}
   V_{j}=\sum_{k=1}^{n} e^{i2\pi/n}H_{k,k}
 \end{equation}
 for each index $j\in\{1,\dots,n\}$. It may be verified that each of the matrices $U_{\ell,a}$ and $V_j$ is unitary. Analogous to the proofs of Theorems \ref{thm:antisymWH} and \ref{thm:symmWHeven}, one has
 \begin{align}
   \frac{1}{n+1}\sum_{\ell=1}^{n}\sum_{a=1}^{\frac{n+1}{2}} U_{\ell,a}XU_{\ell,a}^* 
     & =  
      \frac{1}{n+1} \sum_{\ell=1}^{n}\sum_{a=1}^{\frac{n+1}{2}} \; \sum_{b,c=1}^{\frac{n+1}{2}} 
       2 \; \zeta^{2a(b-c)}F_{\ell,b}XF_{\ell,c}^*  \nonumber\\[1ex]
     &=\sum_{0\leq j <k \leq n} H_{j,k} X H_{j,k}\nonumber\\
     &=\sum_{1\leq j <k \leq n} H_{j,k} X H_{j,k} + \frac{1}{2}\sum_{j=1}^nH_{j,j}X H_{j,j}\label{eq:Uoddsymm}
 \end{align}
 and
 \begin{equation}\label{eq:Voddsymm}
   \frac{1}{2n}\sum_{j=1}^nV_jXV_j^* = \frac{1}{2}\sum_{j=1}^nH_{j,j}X H_{j,j}
 \end{equation}
 for each $X\in\M_n$. Putting together the results of \eqref{eq:Uoddsymm} and \eqref{eq:Voddsymm}, we see that
 \begin{equation}
   \Phi_0(X) = \frac{2}{n+1}\Bigl(\frac{1}{n+1}\sum_{\ell=1}^{n}\sum_{a=1}^{\frac{n+1}{2}} U_{\ell,a}XU_{\ell,a}^* + \frac{1}{2n}\sum_{j=1}^nV_jXV_j^*\Bigr)
 \end{equation}
 holds for each $X\in\M_n$ and thus $\Phi_0$ can be expressed as a convex combination of $n(n+1)/2 + n = n(n+3)/2$ unitary channels. 
\end{proof}

While Theorems \ref{thm:antisymWH} and \ref{thm:symmWHeven} indicate that the symmetric and anti-symmetric Werner--Holevo channels have minimal mixed-unitary rank for all even integers $n$, Theorem \ref{thm:symmWHodd} only gives an upper bound on the mixed-unitary rank of the symmetric Werner--Holevo channel $\Phi_1$ for odd $n$. As it must be the case that the mixed-unitary rank of a channel is at least equal to its rank, the mixed-unitary rank $N$ of $\Phi_1$ for odd $n$ is therefore bounded by
\begin{equation}
 \frac{n(n+1)}{2}\leq N\leq \frac{n(n+3)}{2}.
\end{equation}
It would be interesting if it turned out that $\Phi_1$ were to have minimal mixed-unitary rank for every positive integer $n$. As the Choi representation of the symmetric Werner--Holevo channel is proportional to the projection matrix onto the symmetric subspace of $\complex^n\otimes\complex^n$,
\begin{equation}
 J(\Phi_0) = \frac{2}{n+1} \, \Pi_0 \,,
\end{equation}
finding a minimal mixed-unitary decomposition of $\Phi_0$ for an integer $n$ amounts to finding an orthogonal collection of $n(n+1)/2$ unitary matrices $\{U_1,\dots,U_{n(n+1)/2}\}\subset\U_n$ such that each $U_k$ is symmetric in the sense that $U_k^\t=U_k$. If such a collection could be found, it would satisfy
\begin{equation}
 \Pi_0 = \frac{1}{n} \sum_{k=1}^{\frac{n(n+1)}{2}} \vec{U_k}\vec(U_k)^*,
\end{equation}
as $\op{rank}(\Pi_0) = n(n+1)/2$. In the case when $n=3$, it turns out that $\Phi_0$ indeed has minimal mixed-unitary rank, as will be shown in Theorem \ref{thm:3symmWH} by explicitly constructing a mixed-unitary decomposition.

\begin{theorem}\label{thm:3symmWH}
  Let $\Phi_0:\M_3\rightarrow \M_3$ be the symmetric Werner--Holevo channel on
  $\M_3$.
  It holds that $\Phi_0$ has mixed-unitary rank equal to $6$ and thus has
  minimal mixed-unitary rank.
\end{theorem}

\begin{proof}
  It is evident that the Choi rank of $\Phi_0$ is equal to
  $\op{rank}(J(\Phi_0)) = 6$.
  Define $\alpha$ and $\zeta$ as
  \begin{equation}
    \alpha = \frac{3}{8}+i\frac{ \sqrt{15}}{8}
    \qquad\text{and}\qquad
    \zeta = \exp(2\pi i/3)
 \end{equation}
 and define unitary matrices $U_1,U_2,U_3,U_4,U_5,U_6\in\U_3$ as
 \begin{equation}
   \begin{aligned}
     U_1 &=  \left(
     \begin{array}{ccc}
       1 & 0  & 0  \\
       0  & \zeta & 0  \\
       0  & 0  & \zeta^2 
     \end{array}
     \right)&
     U_2 &= \left(
     \begin{array}{ccc}
       1 & 0  & 0  \\
       0  & \zeta^2 & 0  \\
       0  & 0  & \zeta 
     \end{array}
     \right) &
     U_3&=\left(
     \begin{array}{ccc}
       \frac{1}{2} & -\alpha  & -\alpha  \\
       -\alpha  & \frac{1}{2} & -\alpha  \\
       -\alpha  & -\alpha  & \frac{1}{2} \\
     \end{array}
     \right)\\[2mm]
     U_4&=\left(
     \begin{array}{ccc}
       \frac{1}{2} & \alpha  & -\alpha  \\
       \alpha  & \frac{1}{2} & \alpha  \\
       -\alpha  & \alpha  & \frac{1}{2} 
     \end{array}
     \right)&
     U_5&=\left(
     \begin{array}{ccc}
       \frac{1}{2} & \alpha  & \alpha  \\
       \alpha  & \frac{1}{2} & -\alpha  \\
       \alpha  & -\alpha  & \frac{1}{2} 
     \end{array}
     \right)&
     U_6&=\left(
     \begin{array}{ccc}
       \frac{1}{2} & -\alpha  & \alpha  \\
       -\alpha  & \frac{1}{2} & \alpha  \\
       \alpha  & \alpha  & \frac{1}{2} 
     \end{array}
     \right).
   \end{aligned}
 \end{equation}
 It may be verified that the matrices $U_1,\dots,U_6$ are symmetric, unitary,
 and pairwise orthogonal.
 Hence, $\{U_1,\dots,U_6\}$ is an orthogonal collection of 6 symmetric unitary
 matrices in $\U_3$.
 Comparing this fact with the result of Theorem \ref{thm:wernerholevobases}
 completes the proof.
\end{proof}

The construction for the mixed-unitary decomposition of $\Phi_0$
for $n=3$ presented in the proof of Theorem \ref{thm:3symmWH} does not appear
to generalize for odd integers $n\geq5$.
Nevertheless, numerical evidence seems to suggest that a minimal mixed-unitary
decomposition of $\Phi_0$ might be found for every positive integer.
A proof of the conjecture that $\Phi_0$ has minimal mixed-unitary rank for
every positive integer $n$ would be interesting to pursue.

\section*{Acknowledgements}

The authors are grateful to Mizanur Rahaman and Jamie Sikora for helpful comments and discussions. MG, DL, VP, and JW acknowledge support from Canada's NSERC. MG is partially supported by the Canadian Institute for Advanced Research and through funding provided to IQC by the Government of Canada. DL's Research at Perimeter Institute is supported in part by the Government of Canada through the Department of Innovation, Science and Economic Development Canada and by the Province of Ontario through the Ministry of Colleges and Universities. CKL is supported by Simons Foundation Grant 351047. 

\appendix

\section{A characterization of Schur channels}
\label{appendix:schurmaps}

Two channels $\Phi:\M_n\rightarrow\M_n$ and $\Psi:\M_n\rightarrow\M_n$ are said
to be \emph{unitarily equivalent} if there exist unitary matrices $U,V\in\U_n$
satisfying
\begin{equation}
  \Psi(X) = U\Phi(VXV^*)U^*
\end{equation}
for each $X\in\M_n$.
In this appendix we provide necessary and sufficient conditions that
characterize when a channel is unitarily equivalent to a Schur channel in terms
of the channel's operator system.
We conclude from this characterization that the channels in Examples
\ref{ex:oddprimeWeylchannel} and \ref{ex:IXZmixedunitary} are not equivalent to
Schur channels.

It is known that a channel is a Schur map if and only if every Kraus
representation for the channel consists of only diagonal matrices.
The following lemma provides another useful characterization for Schur
channels.

\begin{lemma}\label{lem:schur}
  A channel $\Phi:\M_n\rightarrow\M_n$ is a Schur map if and only if
  $\Phi(D)=D$ holds for each diagonal matrix $D\in\M_n$.
\end{lemma}

\begin{proof}
  If $\Phi(D)=D$ holds for each diagonal matrix $D\in\M_n$, then each Kraus
  matrix of $\Phi$ must commute with each diagonal matrix, and thus each Kraus
  matrix must itself be diagonal.
  On the other hand, if $\Phi$ is a Schur map then there is a correlation
  matrix $C\in\M_n$ satisfying $\Phi(X) = C\odot X$ for each $X\in\M_n$.
  As each diagonal entry of $C$ must be equal to one, it holds that
  $\Phi(D) = D$ for each diagonal matrix $D\in\M_n$. 
\end{proof}

We now provide a necessary and sufficient condition for characterizing when a
map is unitarily equivalent to a Schur map in terms of the operator system of
the channel.

\begin{theorem}\label{thm:schurequivalent}
  Let $\Phi:\M_n\rightarrow\M_n$ be a channel.
  The following statements are equivalent.
  \begin{enumerate}
  \item The channel $\Phi$ is unitarily equivalent to a Schur map.
  \item The operator system $\mathcal{S}_\Phi$ is a commuting family of
    matrices.
 \end{enumerate}
\end{theorem}

\begin{proof}
  Let $A_1,\dots,A_N\in\M_n$ be linear matrices satisfying
  \begin{equation}
    \Phi(X) = \sum_{k=1}^N A_kXA_k^*
  \end{equation}
  for each $X\in\M_n$.
  First suppose that $\Phi$ is unitarily equivalent to a Schur map.
  There exist unitary matrices $U,V\in\U_n$ such that the channel
  $\Psi:\M_n\rightarrow\M_n$ defined as $\Psi(X) = U\Phi(VXV^*)U^*$ is a Schur
  map.
  The channel $\Psi$ has a Kraus representation of the form
  \begin{equation}
    \Psi(X) = \sum_{j=1}^N (UA_jV)X(UA_jV)^*
  \end{equation}
  and thus $UA_kV$ is a diagonal matrix for each $k\in\{1,\dots,N\}$.
  Moreover, each of the matrices in the collection
  \begin{equation}
    V^*\mathcal{S}_\Phi V = \op{span}\{V^*A_j^*U^*UA_kV\, :\, j,k\in\{1,\dots,N\}\}
  \end{equation}
  is also diagonal.
  It follows that $\mathcal{S}_\Phi$ is a commuting family of normal matrices
  in $\M_n$.
 
  For the other direction, suppose that $\mathcal{S}_\Phi$ is a commuting
  family.
  As $\mathcal{S}_\Phi$ is self-adjoint, each matrix in $\mathcal{S}_\Phi$ is
  also normal.
  There exists a unitary matrix $V\in\U_n$ such that $V^*A_j^*A_kV$ is a
  diagonal matrix for each pair of indices $j,k\in\{1,\dots,N\}$.
  For any two diagonal matrices $D_0,D_1\in\M_n$, one has that 
  \begin{equation}
    \begin{aligned}
      \Phi(VD_0V^*)\Phi(VD_1V^*)  
      & = \sum_{j,k=1}^N A_jVD_0(V^*A_j^*A_kV)D_1V^*A_k^*\\
      & = \sum_{j,k=1}^N A_jVD_0D_1(V^*A_j^*A_kV)V^*A_k^*\\
      & = \Phi(VD_0D_1V^*)\Phi(\I_n) 
    \end{aligned}
  \end{equation}
 as each of the matrices in $\{V^*A_j^*A_kV\,:j,k\in\{1,\dots,N\}\}$ is
 diagonal and commutes with the diagonal matrices $D_0$ and $D_1$.
 Define  matrices $P_1,\dots,P_n\in\M_n$ as 
 \begin{equation}
  P_k = \Phi(VE_{k,k}V^*)
 \end{equation}
 for each $k\in\{1,\dots,n\}$.
 For indices $j,k\in\{1,\dots,n\}$ with $j\neq k$, one has that 
 \begin{equation}
   P_jP_k = \Phi(VE_{j,j}E_{k,k}V^*)\Phi(\I_n)  = 0
 \end{equation}
 as $E_{j,j}E_{k,k}=0$.
 Moreover, it holds that $\Tr(P_k)=\Tr(E_{k,k})=1$ and that $P_k$ is positive
 for each $k\in\{1,\dots,n\}$ as $\Phi$ is a quantum channel.
 The collection $\{P_1,\dots,P_n\}\subset\M_n$ is therefore an orthogonal set
 of positive matrices each with trace equal to 1.
 Hence there must exist a unitary matrix $U\in\U_n$ such that $UP_jU^* =
 E_{j,j}$ for each $j\in\{1,\dots,n\}$.
 Define a channel $\Psi:\M_n\rightarrow\M_n$ as 
 \begin{equation}
  \Psi(X) = U\Phi(VXV^*)U^*
 \end{equation}
 for each $X\in\M_n$.
 From the  observations above, one finds that $\Psi(E_{k,k})=E_{k,k}$ for each
 $k\in\{1,\dots,n\}$, and thus $\Psi(D)=D$ for each diagonal matrix
 $D\in\M_n$.
 It follows that $\Psi$ is a Schur map by Lemma \ref{lem:schur}.
 This completes the proof.
\end{proof}

\begin{remark}
 Every unital quantum channel with Choi rank at most 2 is unitarily equivalent
 to a Schur map.
 This fact was proven in \cite{Landau1993}, but we remark that another proof of
 this fact can be found by making use of Theorem \ref{thm:schurequivalent}.
 Indeed, let $\Phi:\M_n\rightarrow\M_n$ be a unital quantum channel for some
 positive integer $n$ such that $\op{rank}(J(\Phi))\leq2$.
 There exist matrices $A_0,A_1\in\M_n$ such that 
 \begin{equation}
  \Phi(X) = A_0XA_0^* + A_1XA_1^*
 \end{equation}
 is a Kraus representation of $\Phi$.
 As $\Phi$ is unital and trace preserving, these matrices must satisfy
 $A_0^*A_0 + A_1^*A_1  = \I_n$ and $A_0 A_0^* + A_1A_1^*=\I_n$.
 The operator system of $\Phi$ may be given by
 $\mathcal{S}_\Phi=\op{span}\{A_0^*A_0,A_0^*A_1,A_1^*A_0,A_1^*A_1\}$, and it is
 straightforward to verify that each of these matrices commute with one
 another:
 \begin{equation}
   \begin{aligned}
     (A_0^*A_0) (A_0^*A_1) & =
     A_0^*(\I-A_1A_1^*)A_1 = A_0^*A_1(\I-A_1^*A_1)
     = (A_0^*A_1)(A_0^*A_0),\\[1mm]
     (A_0^*A_0) (A_1^*A_0) & =
     (\I-A_1^*A_1)A_1^*A_0 = A_1^*(1-A_1A_1^*)A_0 = (A_1^*A_0)(A_0^*A_0),\\[1mm]
     (A_1^*A_1) (A_0^*A_1) & =
     (\I-A_0^*A_0)A_0^*A_1 = A_0^*(\I-A_0A_0^*)A_1 = A_0^*A_1A_1^*A_1,\\[1mm]
     (A_1^*A_1) (A_1^*A_0) & =
     A_1^*(\I-A_0A_0^*)A_0 = A_1^*A_0(\I-A_0^*A_0) = A_1^*A_0A_1^*A_1,\\[1mm]
     (A_0^*A_0) (A_1^*A_1) & =
     (\I-A_1^*A_1)(\I-A_0^*A_0) = \I - A_0^*A_0 - A_1^*A_1 +A_1^*A_1A_0^*A_0\\
     & = (A_1^*A_1)(A_0^*A_0),\\[1mm]
     (A_0^*A_1) (A_1^*A_0) & =
     A_0^*(\I-A_0A_0^*)A_0 = A_0^*A_0(\I-A_0^*A_0) = (\I-A_1^*A_1)A_1^*A_1\\
     & = A_1^*(\I-A_1A_1^*)A_1 = (A_1^*A_0)(A_0^*A_1).
   \end{aligned}
 \end{equation}
\end{remark}

 
\bibliographystyle{alpha}
\bibliography{mixedunitary}

\end{document}